\documentclass[a4paper,UKenglish,numberwithinsect]{lipics-v2021}

\usepackage{booktabs}
\usepackage{float}
\usepackage{tikz}
\usetikzlibrary{arrows.meta,positioning}
\usepackage{pgfplots}
\usepgfplotslibrary{groupplots}
\pgfplotsset{compat=1.18}

\definecolor{batchrouteblue}{RGB}{34,113,179}
\definecolor{batchroutefill}{RGB}{226,239,250}
\definecolor{batchcellfill}{RGB}{255,235,170}
\definecolor{batchcellline}{RGB}{218,142,0}
\definecolor{batchroutepoint}{RGB}{190,48,48}
\definecolor{cyclecellblue}{RGB}{216,233,247}
\definecolor{cyclelineblue}{RGB}{86,139,173}
\definecolor{cyclewarmfill}{RGB}{255,239,215}
\definecolor{cyclewarmline}{RGB}{230,126,34}
\definecolor{cyclegrayfill}{RGB}{240,242,244}
\definecolor{cycleredline}{RGB}{201,55,48}
\definecolor{cycletextgray}{RGB}{82,91,99}
\definecolor{staircellblue}{RGB}{216,233,247}
\definecolor{stairlineblue}{RGB}{86,139,173}
\definecolor{stairblockgreen}{RGB}{220,244,224}
\definecolor{stairgreenline}{RGB}{38,145,71}
\definecolor{stairblockgray}{RGB}{238,240,242}
\definecolor{stairgrayline}{RGB}{112,120,128}
\definecolor{stairwarmline}{RGB}{230,126,34}
\definecolor{stairredline}{RGB}{201,55,48}
\definecolor{stairtextgray}{RGB}{82,91,99}
\definecolor{benchmarkproposed}{RGB}{215,48,39}
\definecolor{benchmarkbaseline}{RGB}{117,81,166}
\definecolor{benchmarkgrid}{RGB}{205,210,215}

\hideLIPIcs
\makeatletter
\apptocmd{\thebibliography}{\fontsize{9}{11}\selectfont}{}{}
\makeatother

\floatstyle{plain}
\newfloat{algorithm}{tbp}{loa}
\floatname{algorithm}{Algorithm}
\newcounter{algline}
\newenvironment{algorithmic}{%
  \setcounter{algline}{0}%
  \begin{tabular}{@{}r@{\hspace{0.7em}}p{\dimexpr\linewidth-3em\relax}@{}}%
}{\end{tabular}}
\newcommand{\AlgLine}[2]{%
  \stepcounter{algline}%
  \makebox[1.5em][r]{\scriptsize\thealgline} & \hspace*{#1em}#2\\%
}

\newcommand{\floor}[1]{\left\lfloor #1\right\rfloor}
\newcommand{\ceil}[1]{\left\lceil #1\right\rceil}
\newcommand{\MZ}{\mathsf M_{\mathbb Z}}
\DeclareMathOperator{\mult}{\mathcal W}

\title{Counting Lattice Rectangles in \texorpdfstring{$O(n\log n)$}{O(n log n)} Arithmetic Operations}
\titlerunning{Counting Lattice Rectangles in $O(n\log n)$ Operations}

\author{Dmitry Babichev}{Independent researcher, France}{}{}{}
\author{Tatiana Shpakova}{Independent researcher, France}{}{}{}
\authorrunning{Dmitry Babichev and Tatiana Shpakova}
\Copyright{Dmitry Babichev and Tatiana Shpakova}


\modulolinenumbers[5]

\begin{document}
\nolinenumbers
\hypersetup{pdfsubject={Preprint}}
\raggedbottom
\maketitle

\begin{abstract}
Let $F(n)$ be the number of rectangles, not necessarily axis-parallel, whose
vertices belong to the $n\times n$ square grid of lattice points.  We give an
exact algorithm that computes one prescribed value $F(n)$ in
$O(n\log n)$ arithmetic operations and $O(n^{3/4})$ arithmetic words of
working memory.  The algorithm decomposes the count into M\"obius divisor
layers, partitions weighted floor-moment queries by a truncated Euclidean
coefficient-cone recursion, and reuses uniform marker grids along common
coefficient paths.  Each marker requires only its uniform cell and constant-size
corrections at nearby boundaries, which select an exact precompiled cell
operator.  All integer operands have $O(\log n)$ bits.  An exact 128-bit C++
implementation for the reported input range is compared experimentally with the previous
$O(n\log^2 n)$ algorithm.
\end{abstract}

\section{Introduction}

Let $F(n)$ denote the number of rectangles, axis-parallel or oblique, whose
four vertices belong to the point grid $\{0,\ldots,n-1\}^2$.  This is OEIS
A085582~\cite{oeis}.  Previous work by Babichev and Babichev~\cite{paper}
computes one exact value in $O(n\log^2 n)$ arithmetic operations.  We reduce
this bound to $O(n\log n)$ while retaining sublinear working memory.

The problem combines lattice-point enumeration with primitive-direction
counting.  General rational-polyhedral methods are developed by Beck and
Robins~\cite{beck-robins} and Barvinok~\cite{barvinok}; primitive planar
lattice points have also been studied analytically~\cite{huxley-nowak} and
algorithmically~\cite{pawlewicz-patrascu}.  Here the orthogonality and
side-multiplicity structure of rectangles leads to a more specialized
floor-moment problem.

Weighted floor sums and reciprocity also occur in integer-point transforms of
rational cones and in power sums of floor functions, notably through
Dedekind--Carlitz and Rademacher--Carlitz polynomials and generalized Dedekind
sums~\cite{beck-haase-matthews,beck-kohl,brown-floor-powers}; the continued-fraction and
continuant facts used below are standard~\cite{rockett-szusz}.  Reusing initial
Euclidean quotient sequences goes back at least to Lehmer's large-integer GCD
method and its modern fast-arithmetic
descendants~\cite{lehmer-gcd,brent-zimmermann}.  Those methods accelerate one large GCD
instance.  Here common prefixes are instead shared across the
$\Theta(N)$ related coefficient pairs of a layer of size $N$ and are coupled to
two finite-endpoint marker traces.

The first step expresses the non-axis-parallel contribution through primitive
direction vectors and then through M\"obius divisor layers.  For one layer of
size $N$, a half-domain reduction leaves six weighted floor moments.  A
truncated Euclidean recursion groups their parameters by coefficient paths.
Uniform marker grids compile the common part of each path, while a constant
menu of local boundary corrections identifies the exact cell operator for
every query.  Choosing the three internal scales as
$N^{1/2},N^{1/4},N^{1/4}$ gives linear arithmetic work per distinct divisor
layer and $O(N^{3/4})$ words of peak storage.  Summing the distinct quotient
layers yields the claimed $O(n\log n)$ total.

The geometric parametrization, M\"obius divisor-layer identity, and pointwise
six-moment Euclidean kernel build on~\cite{paper}.  We restate those reductions
in the exact form needed here.  The new contribution starts at batch
evaluation: coefficient cones compile common coefficient traces, modular-orbit
marker grids classify both finite endpoints, and constant-size correction
menus select an exact precompiled cell operator.  These ingredients replace
an independent $O(\log N)$ recursion for each coefficient pair by linear work
per layer.  We give the complete algorithm and prove its $O(n\log n)$ time and
$O(n^{3/4})$ memory bounds.  A companion paper studies the complementary
all-values problem of computing the complete prefix $F(1),\ldots,F(N)$ by
rational generating series~\cite{companion-all-values}.

We distinguish arithmetic-operation bounds from bit complexity.  An
\emph{arithmetic word} stores a signed integer of $O(\log n)$ bits, and one
arithmetic operation is an exact addition, subtraction, multiplication,
Euclidean division or remainder, or comparison on such words; random-access
table lookup costs $O(1)$.  GCD and modular-inverse construction costs are
accounted for explicitly where they occur; see also the standard computational
algebra model in~\cite{modern-computer-algebra}.  Lemma~\ref{lem:one-value-operands}
proves that every exact operand fits this model.  Thus the arithmetic bound implies
$O(n\log n\,\MZ(\log n))$ bit operations, where $\MZ(b)$ is the cost of
multiplying two $b$-bit integers.  Fast exact division and remainder have the
same asymptotic cost at this scale~\cite{brent-zimmermann}.  The classical
Sch\"onhage--Strassen bound is $\MZ(b)=O(b\log b\log\log b)$~\cite{schoenhage-strassen};
using the current $\MZ(b)=O(b\log b)$ bound~\cite{harvey-vdh} gives
$O(n\log^2 n\log\log n)$ bit operations.

The paper first derives the geometric and moment reductions, then outlines the
batching architecture and develops its Euclidean state recursion, marker
corrections, cell operators, and complete algorithm.  The experiments follow
the main theorem.  Lengthier identities and proofs are collected in the appendices.

\section{Geometric reduction and M\"obius divisor layers}
\label{sec:geometric-one}\label{sec:mobius-layers}

Every lattice rectangle has two perpendicular side vectors $x(u,v)$ and
$y(-v,u)$, where $x,y$ are positive integers and $(u,v)$ is a primitive
integer vector.
Reflecting and interchanging the sides gives a unique representative with
$u\geqslant v\geqslant0$ and $x\geqslant y$.  Its axis-aligned bounding box
has width $xu+yv$ and height $xv+yu$, hence it has
$(n-xu-yv)(n-xv-yu)$ translations in the $n\times n$ point grid.  For
$0<v<u$, the reflection orbit has size two when $x=y$ and four when $x>y$;
set $\mult(x,y)=2$ when $x=y$ and $\mult(x,y)=4$ when $x>y$.

There are two boundary direction classes.  If $v=0$, primitivity forces
$(u,v)=(1,0)$ and gives the axis-parallel rectangles.  If $u=v$, primitivity
forces $(u,v)=(1,1)$ and gives the rectangles whose sides have slopes $1$ and
$-1$.  Writing $s=x+y$ in this diagonal class, their total is
\begin{equation}\label{eq:boundary-count}
F_0(n)=\binom n2^2+\sum_{s=2}^{n-1}(s-1)(n-s)^2
=\frac{n(n-1)^2(2n-1)}6.
\end{equation}
The remaining representatives have $u>v>0$, so
\begin{equation}\label{eq:geometric}
F_1(n)=
\sum_{\substack{u>v>0,\ x\geqslant y\geqslant1\\
(u,v)=1,\ xu+yv\leqslant n}}
\mult(x,y)(n-xu-yv)(n-xv-yu).
\end{equation}
Here $(xu+yv)-(xv+yu)=(u-v)(x-y)\geqslant0$, so the first
bounding-box dimension dominates the second and the single inequality in
\eqref{eq:geometric} is sufficient.
The preceding classification is exhaustive and its representatives are
disjoint, hence
\begin{equation}\label{eq:full-geometric-count}
F(n)=F_0(n)+F_1(n).
\end{equation}

Use the standard M\"obius identity
$\mathbf 1_{(u,v)=1}=\sum_{d\mid u,\,d\mid v}\mu(d)$
\cite[Chapter~2]{apostol}, where $\mathbf 1_E$ is the indicator of $E$ and
$\mu$ is the M\"obius function.  Write $u=da$, $v=db$, and put
\begin{equation}
N_d=\floor{n/d},\qquad
S_d(n)=
\sum_{\substack{a>b\geqslant1,\ x\geqslant y\geqslant1\\ax+by\leqslant N_d}}
\mult(x,y)
\bigl(n-d(ax+by)\bigr)\bigl(n-d(bx+ay)\bigr).
\end{equation}
Then
\begin{equation}\label{eq:layers}
F_1(n)=\sum_{d\leqslant n}\mu(d)S_d(n).
\end{equation}
For a fixed layer size $N$, let
\[
\Omega_N=
\{(a,b,x,y):a>b\geqslant1,\ x\geqslant y\geqslant1,\ ax+by\leqslant N\}
\]
and define
\begin{align*}
A(N)&=\sum_{\Omega_N}\mult(x,y),\\
B(N)&=\sum_{\Omega_N}\mult(x,y)(a+b)(x+y),\\
C(N)&=\sum_{\Omega_N}\mult(x,y)
\bigl(ab(x^2+y^2)+(a^2+b^2)xy\bigr).
\end{align*}
We call the coefficient triple $(A(N),B(N),C(N))$ the \emph{divisor layer of
size $N$}.
Since $(a-b)(x-y)\geqslant0$, the displayed constraint also enforces the
second bounding-box inequality.  Expanding the two linear factors gives
\begin{equation}\label{eq:layer-coefficients}
S_d(n)=n^2A(N_d)-ndB(N_d)+d^2C(N_d).
\end{equation}
Hence quotient blocking evaluates $(A(N),B(N),C(N))$ once and reuses it for
all divisors with the same $N_d$.

\section{Weighted floor moments}\label{sec:weighted-floor-moments}

This section converts the four-dimensional layer sum into at most one
six-moment floor query per coefficient pair, plus formulas involving ordinary
power sums only.  The half-domain symmetry is the interface between the
geometric counting problem and the recursive Euclidean transducer developed
below.
For background on floor functions and their role in lattice-point
reciprocity, see~\cite{concrete-math,beck-robins}.

For $0\leqslant b,\beta<a$ and $q\geqslant0$, put
$f(t)=\floor{(bt+\beta)/a}$ for $0\leqslant t<q$.
A query with $q=0$ has an empty summation range.  A normalized query with
$b=0$ is also identically zero because $0\leqslant\beta<a$.  Both cases return the
zero moment vector before any reciprocal step.  In the sequel a
\emph{nontrivial query} means $q>0$ and $0<b<a$.  Write
\begin{equation}\label{eq:sixmom}
 \Phi_{ij}(f):=\sum_{0\leqslant t<q}t^i f(t)^j,\qquad
 \boldsymbol\Phi(f):=
 (\Phi_{01},\Phi_{11},\Phi_{21},\Phi_{02},\Phi_{12},\Phi_{03})^{\mathsf T}.
\end{equation}
This six-moment vector is required by the geometric upper-limit formulas in
Appendix~\ref{app:half-domain-moments}.  The Euclidean recurrence, however,
has a simpler closed action on the six lattice moments
\[
 \mathsf L_{ij}(f):=\sum_{0\leqslant t<q}
 \sum_{1\leqslant s\leqslant f(t)}t^is^j,\qquad
 \mathbf L(f):=
 (\mathsf L_{00},\mathsf L_{10},\mathsf L_{20},
  \mathsf L_{01},\mathsf L_{11},\mathsf L_{02})^{\mathsf T},
\]
by the explicit identities
\begin{equation}\label{eq:floor-lattice-conversion}
\begin{gathered}
\Phi_{01}=\mathsf L_{00},\quad \Phi_{11}=\mathsf L_{10},\quad
\Phi_{21}=\mathsf L_{20},\\
\Phi_{02}=2\mathsf L_{01}-\mathsf L_{00},\quad
\Phi_{12}=2\mathsf L_{11}-\mathsf L_{10},\quad
\Phi_{03}=3\mathsf L_{02}-3\mathsf L_{01}+\mathsf L_{00}.
\end{gathered}
\end{equation}
The vector $\mathbf L$ is the internal state
transported by every compiled operator in Section~\ref{sec:recursive-operators};
the algorithm converts $\mathbf L$ to the required $\Phi$-moments only once,
after the root staircase has been evaluated.

Put $D=\floor{\sqrt N}$.  For a fixed direction pair $a>b\geqslant1$ and a
summed side pair $x>y\geqslant1$, define
\[
\mathbf w_{a,b}(x,y)=
\left(
1, (a+b)(x+y),
ab(x^2+y^2)+(a^2+b^2)xy
\right).
\]
This is the contribution to $(A,B,C)$ before the side-multiplicity factor.

\begin{lemma}[Half-domain moment reduction]\label{lem:momentreduction}
The triple $(A(N),B(N),C(N))$ is given by the formulas in
Appendix~\ref{app:half-domain-moments}.
The reduction uses at most one six-moment floor query for every
$a>b\geqslant1$ with $a\leqslant D$, $O(1)$ power-sum evaluations per pair,
and an $O(N)$ diagonal sum.  In particular, one layer contains at most
$D(D-1)/2$ nonempty floor queries and is generated in $O(N)$ operations.
\end{lemma}
\begin{proof}[Proof sketch]
Exchange the direction pair $(a,b)$ with the side pair $(x,y)$ and keep
the half with $a\leqslant x$.  Then $a\leqslant\sqrt N$, which leaves only
$O(N)$ direction pairs.  For a fixed pair, the admissible values of $x$ form
one interval for every $y$.  Reversing the $y$-order turns its upper endpoint
into the single normalized floor query $f(t)=\floor{(bt+\beta)/a}$,
$0\leqslant t<q$.
Its six moments give all degree-two sums over the interval.  The lower
endpoint, the equality boundary $a=x$, and the side diagonal $x=y$ are
ordinary power sums.  The exact interval endpoints, six moment conversions,
and boundary formulas are collected in
Appendix~\ref{app:half-domain-moments}.
\end{proof}

For each coefficient pair, call its nonempty normalized floor query together
with the data needed to recover its contribution to $(A(N),B(N),C(N))$ a
\emph{floor-moment record}.  The formulas in
Appendix~\ref{app:half-domain-moments} completely specify record generation.
There are at most $D(D-1)/2=O(N)$ six-moment queries in one layer.

The six moments are closed under extracting integral parts and under the
reciprocal floor transformation.  A pointwise query therefore costs
$O(\log a)$; evaluating all queries pointwise would cost $O(N\log N)$ for
the layer.
The next section outlines how common quotient prefixes and uniform marker grids
remove this logarithm.  All running times below use the arithmetic-word model
fixed in the introduction.

\section{Idea of the \texorpdfstring{$O(n\log n)$}{O(n log n)} algorithm}
\label{sec:one-value-overview}

For a fixed divisor layer of size $N$, Lemma~\ref{lem:momentreduction}
produces a linear-size family of closely related floor staircases.  The
algorithm shares their Euclidean work at two levels.  Recall that
$D=\floor{\sqrt N}$ is the coefficient cap, and let
$\tau=\Theta(\sqrt D)$ be the path threshold.  We write $P$ for a
coefficient path of length $d$ and $L$ for its number of marker slots.
Nonempty queries with root length $q\leqslant5\tau$ use the pointwise
recursion; the remaining nonempty queries are called \emph{long}.  The indices $i,j$ label exact
marker slots on $P$.

First, coefficient pairs with the same initial Euclidean divisions form one
batch.  Their common coefficient path is identified once, and each original
pair is left with only a small terminal pair $(M,R)$.

Within one batch, the cell operators also share prefixes.  If its $L^2$ cell
operators were compiled independently, replaying a path of length $d$ for
every cell would cost $O(L^2d)$.  Since $d$ can be logarithmic, this would add
one logarithm to the final running time.  Instead, a prefix tree merges cell
traces while they are equal and compiles all $L^2$ operators in $O(L^2)$
work.

The endpoint classification is shared at the batch level as well.  Each batch
receives two uniform grids for the endpoint markers $u$
and $v$.  The exact boundaries move slightly with $(M,R)$, but inside one
uniform cell only three nearby boundaries can matter.  We test those few
boundaries and record a local correction in $\{0,1,2,3\}$.

For two markers, the two uniform cells give a rectangle $\mathcal V_{xy}$,
and the two local corrections $(\alpha_u,\alpha_v)$ select one of at most sixteen
conceptual entries.  That entry identifies the exact rectangle operator.  The
reference implementation evaluates this constant-size map implicitly from
prefix counts and the local comparisons, then indexes the compiled operator
directly; it does not store a separate $4\times4$ table for every rectangle.  The operator
sends the record to a small terminal staircase, which is answered from a
shared table and lifted back to the six required moments.

The grids and their correction data are built once per coefficient path.  Building a separate
marker partition for every one of the $\Theta(D^2)$ coefficient pairs would
cost $O(D^2\tau)=O(N^{5/4})$; sharing them keeps the divisor layer linear.

\begin{figure}[H]
  \centering
  \begin{tikzpicture}[
    x=0.68cm,y=0.68cm,
    every node/.style={font=\normalsize,align=center},
    batchbox/.style={draw=batchrouteblue,rounded corners=2pt,thick,
                    fill=batchroutefill,inner sep=5pt},
    menuarrow/.style={-{Latex[length=1.6mm,width=1.1mm]},
                     line width=0.7pt,draw=batchrouteblue},
    recordarrow/.style={-{Latex[length=1.6mm,width=1.1mm]},
                       line width=0.7pt,draw=batchroutepoint}
  ]
    \node at (3.00,5.48) {uniform marker grid\\for one path $P$};
    \foreach \x in {1,2,3,4,5} {
      \draw[batchrouteblue!65,thin] (\x,0.85)--(\x,4.85);
    }
    \foreach \y in {0.85,1.85,2.85,3.85,4.85} {
      \draw[batchrouteblue!65,thin] (1,\y)--(5,\y);
    }
    \fill[batchcellfill] (2,2.85) rectangle (3,3.85);
    \draw[batchcellline,very thick] (2,2.85) rectangle (3,3.85);
    \node[anchor=south,text=batchcellline,fill=white,inner sep=1pt,
          font=\small] at (2.50,3.90)
      {$\mathcal V_{xy}$};

    \foreach \d in {0.25,0.50,0.75} {
      \draw[batchcellline,dashed] ({2+\d},2.85)--({2+\d},3.85);
      \draw[cyclewarmline,dashed] (2,{2.85+\d})--(3,{2.85+\d});
    }
    \coordinate (recordpoint) at (2.625,3.475);
    \fill[batchroutepoint] (recordpoint) circle (2.5pt);
    \node[anchor=south,text=batchroutepoint,font=\small,
          fill=white,inner sep=1.2pt]
      at (6.35,3.98) {one record};

    \draw[-{Latex[length=1.5mm,width=1.0mm]},thick,draw=batchrouteblue]
      (1,0.48)--(5.35,0.48);
    \node[anchor=west] at (5.42,0.48) {$u$};
    \draw[-{Latex[length=1.5mm,width=1.0mm]},thick,draw=batchrouteblue]
      (0.62,0.85)--(0.62,5.10);
    \node[anchor=east] at (0.50,4.90) {$v$};

    \node[anchor=south,align=center] at (9.96,4.80)
      {local correction menu\\[2pt]at most $16$ entries};
    \foreach \i in {0,1,2,3} {
      \foreach \j in {0,1,2,3} {
        \fill[batchroutefill] ({9.00+0.48*\i},{2.72+0.48*\j})
          rectangle +(0.48,0.48);
        \draw[batchrouteblue!75,thin] ({9.00+0.48*\i},{2.72+0.48*\j})
          rectangle +(0.48,0.48);
      }
    }
    \fill[batchcellfill] ({9.00+0.48*2},{2.72+0.48*2})
      rectangle +(0.48,0.48);
    \draw[batchcellline,very thick] ({9.00+0.48*2},{2.72+0.48*2})
      rectangle +(0.48,0.48);
    \node[anchor=north] at (9.96,2.62) {$\alpha_u$};
    \node[anchor=east] at (8.80,4.22) {$\alpha_v$};
    \node[batchbox,minimum width=2.90cm,minimum height=1.15cm]
      (operator) at (14.40,3.92)
      {compiled operator\\$\mathsf{Op}_{P;i,j}$};
    \node[batchbox,minimum width=3.05cm,minimum height=1.15cm]
      (terminal) at (14.40,1.62)
      {terminal lookup\\and operator lift};
    \node[batchbox,minimum width=1.78cm,minimum height=1.15cm]
      (moments) at (19.00,1.62) {six root\\moments};

    \draw[recordarrow,shorten >=-0.7pt]
      (recordpoint) -- (9.00,3.92);
    \node[anchor=north,text=batchroutepoint,font=\small]
      at (7.00,3.04)
      {local corrections\\$(\alpha_u,\alpha_v)=(2,2)$};
    \draw[menuarrow,shorten >=-0.7pt] (10.92,3.92)--(operator.west);
    \draw[menuarrow,shorten >=-0.7pt] (operator.south)--(terminal.north);
    \draw[menuarrow,shorten >=-0.7pt] (terminal.east)--(moments.west);
  \end{tikzpicture}
  \caption{A record first enters the uniform rectangle
  $\mathcal V_{xy}$.  Only the nearby dashed true boundaries matter inside
  that rectangle; they give the local corrections $(\alpha_u,\alpha_v)$.  The selected
  conceptual menu entry identifies the exact compiled operator.}
  \label{fig:query-batching}
\end{figure}
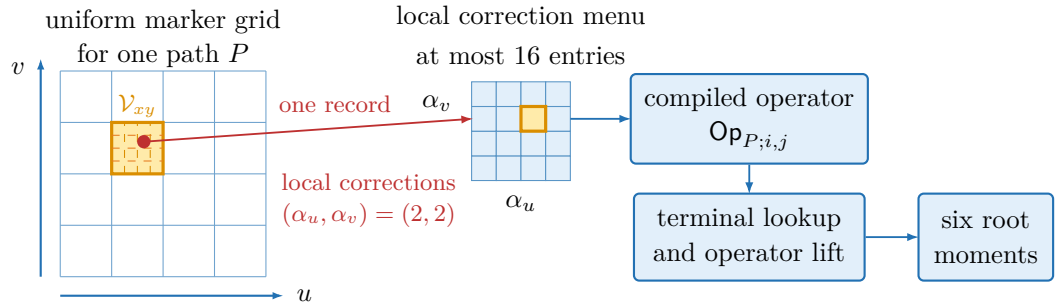

Four facts give linear work in one divisor layer: the coefficient batches
partition all direction pairs; each path has only quadratically many cell
operators and correction-menu entries; their total number over all paths
is linear in the layer size; and every long record performs constant work.
Quotient blocking over the divisor layers then gives the claimed
$O(n\log n)$ running time.

The remaining construction is organized as follows.
Section~\ref{sec:normalized-state} derives the exact reciprocal rule for one
query.  Section~\ref{sec:terminal-coordinates} defines coefficient batches and
terminal coordinates.  Section~\ref{sec:marker-orbit} defines the uniform-grid
corrections, and Section~\ref{sec:recursive-operators} compiles the true
operators and specifies the correction map.  The technical proofs are
collected in Appendix~\ref{app:part1-batching}; the marker-grid construction
is detailed in Appendix~\ref{app:marker-orbit-proof}.

\section{One-query Euclidean cycle and affine lift}\label{sec:normalized-state}

Section~\ref{sec:weighted-floor-moments} reduces one divisor layer to $O(N)$
floor-moment queries.  Evaluating every query by its own Euclidean recursion
would cost $O(N\log N)$.  This section derives the local reciprocal cycle for
one query; Sections~\ref{sec:terminal-coordinates}--\ref{sec:recursive-operators}
then compile and reuse it across a batch.  A normalized query has four
independent parameters, but the recurrence becomes affine after two endpoint
data are made explicit.  The direction pair
$(a,b)$ determines the ordinary Euclidean quotients, whereas the intercept
$\beta$ and the interval length $q$ determine where the finite floor staircase
begins and ends.

\subsection{The staircase and its two boundary markers}

After the zero cases above have been removed, every nontrivial floor query
has the normalized
form
\begin{equation}\label{eq:floor-query}
f(t)=\floor{\frac{bt+\beta}{a}},\qquad 0\leqslant t<q,\qquad
q>0,\quad 0<b<a,\quad 0\leqslant\beta<a.
\end{equation}
Its lattice staircase is
\begin{equation}\label{eq:lattice-staircase}
\Lambda_f=\{(t,s)\in\mathbb Z^2:0\leqslant t<q,\ 1\leqslant s\leqslant f(t)\}.
\end{equation}
Its graph is a staircase inside a rectangle of width $q$.  Its height is
$h=f(q-1)=\floor{(b(q-1)+\beta)/a}$.
The standard reciprocal transformation counts the complementary staircase
after interchanging the two coordinate directions.  The transformed floor
function is
\begin{equation}\label{eq:reciprocal-query}
\widehat f(k)=\floor{\frac{ak+u}{b}},\qquad
0\leqslant k<h,\qquad u=a-\beta-1.
\end{equation}
Thus $u$ records the left boundary of the complementary staircase.  At the
right boundary, write the endpoint numerator as $b(q-1)+\beta=ah+v$.
The remainder $v$ records where the final step meets that boundary.

\begin{definition}[Affine query state]
\label{def:affine-query-state}
The \emph{affine query state} of the normalized floor query
\eqref{eq:floor-query} is the six-coordinate column vector
\begin{equation}\label{eq:state}
\mathbf s(a,b,q,\beta):=(a,b;q,h;u,v)^{\mathsf T},\qquad
u=a-\beta-1,\quad v=b(q-1)+\beta-ah.
\end{equation}
The semicolons group its coordinates as
\[
\underbrace{(a,b)}_{\text{direction coefficients}}\ ;\quad
\underbrace{(q,h)}_{\text{staircase dimensions}}\ ;\quad
\underbrace{(u,v)}_{\text{boundary markers}}.
\]
For a sequence of reciprocal cycles, $\mathbf s_{\mathrm{root}}$ denotes its
initial state and $\mathbf s_{\mathrm{term}}$ the state after the last cycle.
\end{definition}

\begin{definition}[Boundary markers and quotient trace]
\label{def:boundary-markers}
For a normalized finite staircase, the remainders $u$ and $v$ in
\eqref{eq:state} are its \emph{left} and \emph{right boundary markers}.  They
record where the complementary staircase starts and where the final step
meets the opposite boundary.  At reciprocal cycle $i$, write the current state
as $(a_i,b_i;q_i,h_i;u_i,v_i)$ and define
$A_i=\floor{a_i/b_i}$, $U_i=\floor{u_i/b_i}$, and
$V_i=\floor{v_i/b_i}$.  These are respectively the coefficient, left-marker,
and right-marker quotients.  A \emph{word} here is simply a finite sequence of
quotients.  Thus
$(A_i)_i$ is the \emph{coefficient word}, while $(U_i)_i$ and $(V_i)_i$ are
the two \emph{marker words}.  The sequence of quotient triples
$((A_i,U_i,V_i))_i$ is the \emph{complete quotient trace} of the query.
\end{definition}

The coefficient word depends only on $(a,b)$.  The marker words additionally
depend on the finite endpoints of the staircase.  A shared operator is
therefore indexed by the coefficient word and both marker words.

Although $(a,b,q,\beta)$ determines the query and $\beta=a-u-1$, the lifted
state retains $h$, because it becomes the next interval length, and retains
$u,v$, because $U=\floor{u/b}$ and $V=\floor{v/b}$ are classified separately.
We therefore retain the affine lift throughout.  Its bounds are
$0<b<a$, $0\leqslant u,v<a$, $q>0$, and $h\geqslant0$.
If $h=0$,
monotonicity of $f$ shows that all its values are zero; such a state is
terminal and is never used as the parent of another reciprocal cycle.

The affine lift describes one finite staircase, whereas the six-moment vector
in~\eqref{eq:sixmom} contains six aggregate sums over that staircase.  A compiled operator
updates the state affinely and transports the moment vector linearly, with a
polynomial boundary correction.

\subsection{One complete Euclidean cycle}

Assume in this subsection that $h>0$, so another cycle is required.  The
reciprocal query~\eqref{eq:reciprocal-query} is not yet normalized because
its coefficient $a$ can exceed its modulus $b$.  Extract the three quotients
$A=\floor{a/b}$, $U=\floor{u/b}$, and $V=\floor{v/b}$.  Put
$b'=a-Ab$ and $\beta'=u-Ub$.  Then the normalized child is
$f'(k)=\floor{(b'k+\beta')/b}$, and
$\widehat f(k)=Ak+U+f'(k)$.
Write the current and child states explicitly as
\[
 \mathbf s=(a,b;q,h;u,v)^{\mathsf T}
 \xrightarrow{\ \mathcal T_{A,U,V}\ }
 \mathbf s'=(a',b';q',h';u',v')^{\mathsf T}.
\]
The reciprocal transformation followed by normalization is the coordinate
form of this affine state transition:
\begin{align}
a'&=b,& b'&=a-Ab,\nonumber\\
q'&=h,& h'&=q-Ah-(U+V+2-A),\label{eq:cycle}\\
u'&=(U+1)b-u-1,&v'&=(V+1)b-v-1.\nonumber
\end{align}
Its direct derivation from the reciprocal floor function is given in
Appendix~\ref{app:cycle-derivation}.

Equation~\eqref{eq:cycle} is the payoff for the lift: once the quotient triple
$(A,U,V)$ is fixed, one complete cycle is a fixed integral affine map.
The reciprocal and polynomial moment identities are fixed as well.  Hence a
whole common quotient trace can be composed in advance into
one constant-size operator and then applied to every query in the
corresponding piece of the parameter space.

\begin{figure}[!tbp]
  \centering
  \resizebox{1.0\textwidth}{!}{%
  \begin{tikzpicture}[
    x=1cm,y=1cm,
    every node/.style={font=\normalsize,align=center},
    cycleaxis/.style={-{Latex[length=2.2mm]},thick},
    cycletransfer/.style={-{Latex[length=2.6mm]},thick,draw=cyclewarmline}
  ]
    \path[use as bounding box] (0,0) rectangle (13.72,6.75);
    \node at (3.10,6.40) {(a) one normalized query};
    \node at (3.10,5.98) {$f(t)=\floor{(8t+3)/17}$};
    \node[text=cycletextgray] at (3.10,5.56)
      {$0\leq t<18$, $\mathbf s=(17,8;18,8;13,3)^{\mathsf T}$};
    \begin{scope}[shift={(0.85,1.65)}]
      \foreach \x in {0,...,17}{
        \pgfmathtruncatemacro{\hh}{floor((8*\x+3)/17)}
        \foreach \rr in {0,...,7}{
          \ifnum\rr<\hh
            \fill[cyclecellblue] ({0.25*\x},{0.25*\rr}) rectangle
              ({0.25*(\x+1)},{0.25*(\rr+1)});
          \else
            \fill[cyclegrayfill] ({0.25*\x},{0.25*\rr}) rectangle
              ({0.25*(\x+1)},{0.25*(\rr+1)});
          \fi
        }
      }
      \draw[cyclelineblue!55,xstep=0.25cm,ystep=0.25cm] (0,0) grid (4.50,2.00);
      \draw[thick] (0,0) rectangle (4.50,2.00);
      \foreach \pp in {(0.625,0.125),(1.125,0.375),(1.625,0.625),(2.375,0.875),
                       (2.875,1.125),(3.375,1.375),(3.875,1.625),(4.375,1.875)}
        \fill[cyclewarmline] \pp circle (1.7pt);
      \draw[cycleaxis] (0,-0.28)--(4.86,-0.28) node[right] {$t$};
      \draw[cycleaxis] (-0.28,0)--(-0.28,2.36) node[above] {$s$};
      \node[below] at (2.25,-0.46) {$q=18$};
    \end{scope}
    \draw[cycletransfer] (5.85,3.20)--(8.75,3.20);
    \node[above,text width=3.4cm] at (7.30,3.35)
      {$\mathbf s\mapsto\mathbf s'$\\transpose the complement};
    \node[below,text width=3.4cm] at (7.30,3.05)
      {extract\\$(A,U,V)=(2,1,0)$};
    \node[below,text=cycletextgray] at (7.30,2.12)
      {$\widehat f(k)=2k+1+f'(k)$};
    \node at (10.80,6.40) {(b) reciprocal query};
    \node at (10.80,5.98) {$\widehat f(k)=\floor{(17k+13)/8}$};
    \node[text=cycletextgray] at (10.80,5.42)
      {$\mathbf s'=(8,1;8,1;2,4)^{\mathsf T}$};
    \begin{scope}[shift={(9.65,1.10)}]
      \foreach \XX in {0,...,7}{
        \pgfmathtruncatemacro{\base}{2*\XX+1}
        \pgfmathtruncatemacro{\child}{floor((\XX+5)/8)}
        \pgfmathtruncatemacro{\total}{\base+\child}
        \foreach \rr in {0,...,15}{
          \ifnum\rr<\base
            \fill[cyclewarmfill] ({0.25*\XX},{0.25*\rr}) rectangle
              ({0.25*(\XX+1)},{0.25*(\rr+1)});
          \else
            \ifnum\rr<\total
              \fill[cyclecellblue] ({0.25*\XX},{0.25*\rr}) rectangle
                ({0.25*(\XX+1)},{0.25*(\rr+1)});
            \else
              \fill[cyclegrayfill] ({0.25*\XX},{0.25*\rr}) rectangle
                ({0.25*(\XX+1)},{0.25*(\rr+1)});
            \fi
          \fi
        }
      }
      \draw[cyclelineblue!55,xstep=0.25cm,ystep=0.25cm] (0,0) grid (2.00,4.00);
      \draw[thick] (0,0) rectangle (2.00,4.00);
      \draw[cycleredline,thick]
        (0,0)--(0,0.25)--(0.25,0.25)--(0.25,0.75)--
        (0.50,0.75)--(0.50,1.25)--(0.75,1.25)--(0.75,2.00)--
        (1.00,2.00)--(1.00,2.50)--(1.25,2.50)--(1.25,3.00)--
        (1.50,3.00)--(1.50,3.50)--(1.75,3.50)--(1.75,4.00)--
        (2.00,4.00);
      \draw[cycleaxis] (0,-0.28)--(2.36,-0.28) node[right] {$k$};
      \draw[cycleaxis] (-0.28,0)--(-0.28,4.12);
      \node[left] at (-0.28,2.00) {$\widehat f$};
    \end{scope}
    \node[anchor=west] at (8.90,0.50)
      {\tikz\fill[cyclewarmfill,draw=cyclewarmline]
        (0,0) rectangle (0.28,0.20);\ affine part $2k+1$};
    \node[anchor=west] at (8.90,0.14)
      {\tikz\fill[cyclecellblue,draw=cyclelineblue]
        (0,0) rectangle (0.28,0.20);\ child $f'(k)$};
  \end{tikzpicture}%
  }
  \caption{One reciprocal cycle for a normalized query.  All displayed cells
  are equal squares.  Left: $\Lambda_f$ (blue) and its complement (gray).
  Right: the transposed complement splits into the affine part (orange) and
  the recursive child (blue); the red staircase is the graph of $\widehat f$.}
  \label{fig:one-cycle}
\end{figure}
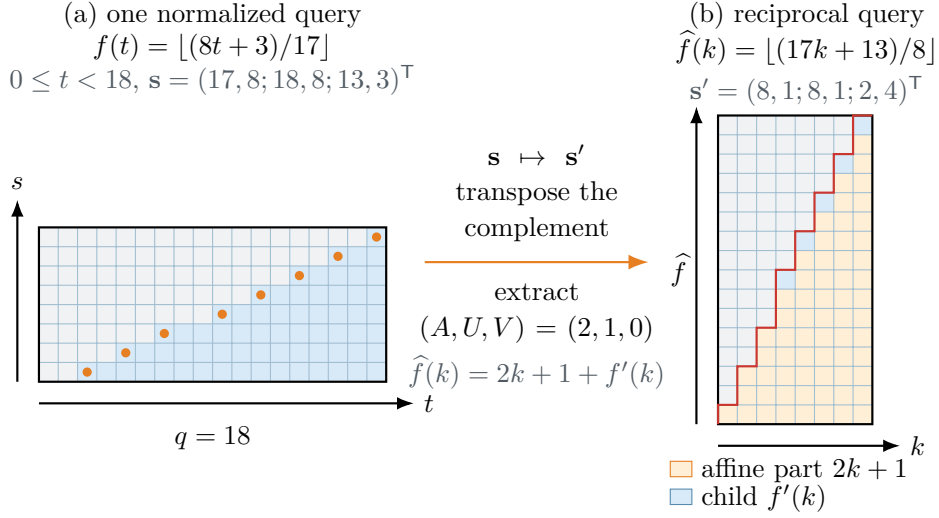

In the left panel of Figure~\ref{fig:one-cycle}, the orange point in row $s$
marks $t=\widehat f(s-1)+1$, the first admitted column.  In the right panel,
transposition gives $\widehat f(k)=2k+1+f'(k)$ and the state transition
$\mathbf s\mapsto\mathbf s'=(8,1;8,1;2,4)^{\mathsf T}$.  This is the local
identity from which an operator is built.  The construction of true marker
rectangles, uniform-grid correction maps, and their operators is given in
Sections~\ref{sec:marker-orbit} and~\ref{sec:recursive-operators}.

\subsection{Why the one-query rule can be reused}

For the displayed query $(a,b,q,\beta)=(17,8,18,3)$, the lift is
$(17,8;18,8;13,3)$ and $(A,U,V)=(2,1,0)$.  Changing only the intercept to
$\beta=4$ gives the different lift $(17,8;18,8;12,4)$ but the same quotient triple.
Equation~\eqref{eq:cycle} therefore applies the same affine state update to
both queries; only the numerical marker values differ.  More generally, once
the complete quotient trace $((A_i,U_i,V_i))_i$ is fixed, its cycle maps can be
composed once.  The paired construction attaches that operator to every true
rectangle on which the sequence is fixed, and the local correction menu
selects it from the two marker codes.

\section{Coefficient-cone recursion and terminal coordinates}
\label{sec:terminal-coordinates}

Fix one layer and put
\begin{equation}\label{eq:layer-scales}
 D=\floor{\sqrt N},\qquad
 \sigma=\ceil{\sqrt D},\qquad
 \tau=\ceil{3\sigma/8},\qquad
 \kappa=\max\!\left(\tau,\ceil{D/\tau}\right).
\end{equation}

\begin{definition}[Layer and batching scales]
\label{def:layer-scales}
$D$ is the largest coefficient occurring in the layer.  The auxiliary scale
$\sigma$ is its square-root scale.  The \emph{path threshold} $\tau$ limits
the common Euclidean prefix compiled for one coefficient batch.  The
\emph{terminal threshold} $\kappa$ limits the moduli stored in the shared
small-modulus table of Section~\ref{sec:periodic-table}.  Since
$\kappa\geqslant\ceil{D/\tau}$, we have $\tau\kappa\geqslant D$, while both
$\tau$ and $\kappa$ are $\Theta(\sqrt D)$.  The product inequality guarantees
  that when a large pair with $R>0$ cannot extend its compiled path, one
  further Euclidean step reaches a modulus below $\kappa$; the case $R=0$ is
  already terminal.  This is formalized in
Lemma~\ref{lem:onestep}.
\end{definition}

One Euclidean division $a_i=A_i b_i+r_i$ replaces $(a_i,b_i)$ by
$(b_i,r_i)$.  A coefficient path of length $d$ consists of $d$ such
divisions, with quotients $A_0,\ldots,A_{d-1}$.  The matrix
\[
 Q(A):=\begin{pmatrix}A&1\\1&0\end{pmatrix},\qquad
 \binom{a_i}{b_i}=Q(A_i)\binom{b_i}{r_i},
\]
reverses one step.  For the full length-$d$ prefix define $P_0=I$,
$P_{i+1}=P_iQ(A_i)$, and $P=P_d$.
Thus $P_i$ maps the coefficient pair after $i$ divisions back to the root
pair, and $P=(\pi_{ij})_{1\leqslant i,j\leqslant2}$ is the \emph{path
matrix} of the full prefix.

\begin{definition}[Coefficient path and cone]
The coefficient word $A_0,\ldots,A_{d-1}$ is an \emph{accepted coefficient
path} if every entry of its path matrix $P$ is at most $\tau$.  Its
\emph{coefficient cone}
is the set of direction pairs $(a,b)$ whose canonical Euclidean divisions begin
with this word.  After these divisions the remaining pair $(M,R)$ is called
the \emph{terminal pair}, or the terminal coordinates of $(a,b)$ on the path.
The root and terminal pairs satisfy
\begin{equation}\label{eq:terminal-regions}
 \binom ab=P\binom MR,
 \qquad M>R\geqslant0,
 \qquad 1\leqslant b<a\leqslant D.
\end{equation}
\end{definition}

The generator starts with the empty path $P=I$.  At a stored path it tries
the canonical next quotient $A=\floor{M/R}$, whenever $R>0$.  The child
$PQ(A)$ is accepted exactly when its four entries remain at most $\tau$.
Geometrically this splits one rational cone in the coefficient triangle into
subcones with a common longer prefix.

\begin{definition}[Stopping rule and coefficient batch]
For an accepted path $P$ define
$\operatorname{ext}(P)=\max\{A\geqslant1:
\max_{i,j}(PQ(A))_{ij}\leqslant\tau\}$, with value zero if no child is
accepted.  A root pair stops at its current
path when $M\leqslant\tau$, $R=0$, or $R>0$ and
$\floor{M/R}>\operatorname{ext}(P)$.  All root pairs that stop at the same path
form one \emph{coefficient batch}.  The path matrix and
coefficient word are shared by the whole batch; only $(M,R)$ varies.
\end{definition}

For direct implementation the maximum in this definition is the explicit
integer expression
\begin{equation}\label{eq:extension-limit-explicit}
 \operatorname{ext}(P)=
 \min\!\left\{
 \floor{\frac{\tau-\pi_{12}}{\pi_{11}}},
 \floor{\frac{\tau-\pi_{22}}{\pi_{21}}}
 \right\},
\end{equation}
where the second term is omitted when $\pi_{21}=0$.  Indeed, the only entries
of $PQ(A)$ that depend on $A$ are
$\pi_{11}A+\pi_{12}$ and $\pi_{21}A+\pi_{22}$; the other two are already
entries of the accepted matrix $P$.  Thus no search over candidate quotients
is required.

For direct integer enumeration at a stored node, let $A_{\rm last}$ be the
incoming quotient, absent at the root.  The stopped terminal pairs are exactly
the following disjoint pieces:
\begin{enumerate}
\item $M\leqslant\tau$, with either $P=I$ or
      $A_{\rm last}M+R>\tau$;
\item at a nonroot node, $M>\tau$, $R=0$, and
      $A_{\rm last}\geqslant2$;
\item $M>\tau$, $R>0$, and
      $M\geqslant(\operatorname{ext}(P)+1)R$.
\end{enumerate}
The predecessor guard in the first piece enforces that the recursive parent
was active.  The condition in the second piece is the standard convention for
the finite Euclidean algorithm: a zero remainder is never represented by an artificial
trailing quotient one.  Denote this stopped set by $\mathcal T(P)$.

\begin{lemma}[One-step terminal]\label{lem:onestep}
At a stopped coefficient cone, either $R=0$, or $M\leqslant\kappa$, or one
further ordinary Euclidean step reaches a positive modulus below $\kappa$.
\end{lemma}
The short inequality proof is given in
Appendix~\ref{app:coefficient-cones}.

\begin{lemma}[Cone partition and enumeration]\label{lem:terminal-bijection}
The map $(P,M,R)\longmapsto(a,b)=P(M,R)$, with
$(M,R)\in\mathcal T(P)$, is a bijection onto the $D(D-1)/2$ pairs
$1\leqslant b<a\leqslant D$.  The recursion has $O(\tau^2)$ matrices, and all
stopped pairs can be enumerated in $O(D^2+\tau^4)$ arithmetic operations.
\end{lemma}
The bijection follows by stopping each canonical Euclidean path at its first
failed continuation.  The matrix count and the explicit integer intervals
used for enumeration are proved in Appendix~\ref{app:coefficient-cones}.

\section{Uniform marker grids and local corrections}
\label{sec:marker-orbit}

Fix one coefficient path $P$, one stopped terminal pair $(M,R)$, and its root
pair $(a,b)=P(M,R)$.  Put $H=a$.  As a marker $t$ ranges over $[0,H)$, its
complete quotient trace changes only at finitely many boundary values.
Between consecutive boundaries the same compiled operator applies.  Thus the
two markers $u$ and $v$ divide their square into rectangles on which one
operator is valid.

\begin{definition}[Boundary multiset and true marker rectangles]
\label{def:marker-intervals}
Let $L$ be the number of one-dimensional marker slots for the current
coefficient path.  Their ordered boundaries form the nondecreasing list
$0=e_0\leqslant e_1\leqslant\cdots\leqslant e_L=H$.
Repetitions are retained, so the same $L$ is used throughout the coefficient
batch even when some boundaries coincide.  The true marker slots are
$I_j=[e_j,e_{j+1})$ for $0\leqslant j<L$, with product rectangles
$\mathcal C_{ij}:=I_i\times I_j$.
A repeated boundary gives an empty slot.  Every nonempty $I_j$ is a maximal
interval on which the marker trace, and hence the compiled operator, is fixed.
\end{definition}

For a numerical marker $t$, let
\begin{equation}\label{eq:marker-slot}
 \iota(t)=\max\{j:0\leqslant j<L,\ e_j\leqslant t\}
\end{equation}
be its true slot.

The true boundaries depend on $(M,R)$, so constructing this partition anew
for every stopped pair would be too expensive.  We start instead with two
uniform grids $\mathcal G_0(P)$ and $\mathcal G_1(P)$ of $[0,H)$, each with
$O(\tau)$ intervals.  The first grid agrees with the true boundaries at
$R=0$, and the second agrees with them at the limiting endpoint $R=M$.  For a given terminal pair
we choose the closer reference: $\nu=0$ when $2R\leqslant M$ and $\nu=1$
otherwise.  Every true boundary then moves by less than one uniform cell.

The true boundaries are not exactly the grid lines, but only a few of them can
affect the answer inside one uniform cell.  The grid construction attaches
those candidate boundaries to the cell.  If $x$ is the cell of $t$ in the
selected grid $\nu$, let $\alpha_\nu(t)$ be the number of its candidates lying at or
before $t$.  Computing $\alpha_\nu(t)$ is the small local correction to the
uniform partition.

\begin{definition}[Uniform cell and local correction]
\label{def:virtual-marker-codes}
The \emph{marker code} $\operatorname{code}_\nu(t):=(x,\alpha_\nu(t))$
consists of the uniform cell $x$ and the local correction
$\alpha_\nu(t)$.
\end{definition}

For two markers, write $\alpha_u:=\alpha_\nu(u)$ and
$\alpha_v:=\alpha_\nu(v)$.  Their uniform cells $(x,y)$ form a uniform
rectangle $\mathcal V_{xy}$.  Each marker has at most four possible corrections, so
$\mathcal V_{xy}$ needs at most $4\times4$ entries.  Entry
$(\alpha_u,\alpha_v)$ points to the exact true rectangle
$\mathcal C_{\iota(u),\iota(v)}$, its compiled operator, and its terminal data.
Denote this menu by $\mathcal M_{\nu,P}(x,y)$.
Figure~\ref{fig:query-batching} shows this single local correction step.

\begin{theorem}[Uniform-grid correction]
\label{thm:virtual-rectangle-menu}
For every accepted path, the two uniform grids and their conceptual local correction
menus have $O(\tau^2)$ entries and can be constructed in $O(\tau^2)$ time.
Inside a uniform cell, at most three true boundaries need to be tested, so
$\alpha_\nu(t)\in\{0,1,2,3\}$.  The marker code determines the exact true
slot~\eqref{eq:marker-slot}; hence two marker codes select the exact true
rectangle and its operator in constant time.
\end{theorem}

Appendix~\ref{app:marker-orbit-proof} gives the grid construction, the
explicit correction formula, and the proof.  The main algorithm needs only
the resulting constant-size menu.

\section{Cell operators and small terminal base cases}
\label{sec:recursive-operators}

For every nonempty true slot, Appendix~\ref{app:marker-orbit-proof} supplies
its marker trace.  Pairing two slot lists then composes the reciprocal
operator for every true rectangle.  The first two subsections compile these
cell operators; the third explains how the remaining small terminal staircase
is answered, and the final subsection assembles the evaluation of one record.
We first record the constant-size algebraic form of one cell operator.

\subsection{Operator attached to one true marker rectangle}

Use the lattice staircase $\Lambda_f$ from~\eqref{eq:lattice-staircase} and its
six degree-two moments
\[
\mathsf L_{ij}(f)=\sum_{(t,s)\in\Lambda_f}t^is^j,
\qquad
(i,j)\in\mathscr D_2:=\{(0,0),(1,0),(2,0),(0,1),(1,1),(0,2)\}.
\]
This is the coordinate order fixed in the definition of $\mathbf L(f)$ in
Section~\ref{sec:weighted-floor-moments}.
Their exact integral conversion to the floor-power moments is
\eqref{eq:floor-lattice-conversion}; the reverse direction follows by summing
$1,s,s^2$ over $1\leqslant s\leqslant f(t)$.

Recall from Section~\ref{sec:normalized-state} that one normalized reciprocal
step writes $\widehat f(k)=Ak+U+f'(k)$, where $f'$ is the child floor query.
For $z\geqslant0$ write
$\operatorname{Pow}_i(z)=\sum_{t=0}^{z}t^i$.  The exact one-step lattice
identity is as follows.  Put $k=s-1$.  For
$0\leqslant k<h$, the condition $s\leqslant f(t)$ is equivalent to
$t\geqslant\widehat f(k)+1=Ak+U+f'(k)+1$.
Consequently, for every $(i,j)\in\mathscr D_2$,
\begin{equation}
\mathsf L_{ij}(f)=\sum_{k=0}^{h-1}\!\left(
\operatorname{Pow}_i(q-1)-\operatorname{Pow}_i(Ak+U)
-\sum_{r=1}^{f'(k)}(Ak+r+U)^i\right)(k+1)^j.
\label{eq:one-step-lattice}
\end{equation}
The inner sum is empty when $f'(k)=0$.
Thus the child coordinates $(k,r)$ enter the root coordinates through
\begin{equation}\label{eq:affinecoords}
t=Ak+r+U,\qquad s=k+1.
\end{equation}
The first two terms inside the parentheses in~\eqref{eq:one-step-lattice}
are explicit power sums; the inner sum is the affine image of the child with
a minus sign.  Thus the child-to-root moment map is the signed action induced
by the affine coordinate map, and~\eqref{eq:one-step-lattice} completely
specifies its one-step action on all six moments.

Because the lengths $(q,h)$ vary across a rectangle, the explicit part of
\eqref{eq:one-step-lattice} is retained as a degree-four polynomial in those
two lengths.  This fixed polynomial space, the six moment coordinates, and
the affine state map are all closed under composition.  The exact monomial
basis, coefficient counts, and common integral scaling are recorded in
Appendix~\ref{app:operator-representation}.

\begin{definition}[Operator of a true marker rectangle]
\label{def:rectangle-operator}
Fix an accepted coefficient path $P$ and a true marker rectangle
$\mathcal C_{ij}=I_i\times I_j$.  The quotient triple at every step is then
fixed.  Their composition is the \emph{rectangle operator}
$\mathsf{Op}_{P;i,j}:=(\mathsf A_{P;i,j},\mathsf U_{P;i,j},\Pi_{P;i,j})$.
It depends only on the path $P$ and the rectangle indices $(i,j)$, equivalently
on the fixed quotient trace.  Its stored coefficients do not depend on the
numerical terminal pair $(M,R)$, the marker values $(u,v)$, or the root lengths
$(q,h)$.  Here $\mathsf A_{P;i,j}$ is the $7\times7$ homogeneous integral
matrix mapping the affine query state of
Definition~\ref{def:affine-query-state}, augmented by a constant coordinate,
from root to terminal; $\mathsf U_{P;i,j}$ is the
$6\times6$ matrix transporting the six lattice moments; and
$\Pi_{P;i,j}(q,h)$ is the six-component polynomial boundary correction.
\end{definition}

\begin{lemma}[Operator closure]\label{lem:closure}
For every fixed coefficient prefix and marker rectangle, let
$\mathbf s_{\mathrm{root}}$ and $\mathbf s_{\mathrm{term}}$ be the initial and
final affine query states of Definition~\ref{def:affine-query-state}.  Then
the terminal state and all six root lattice moments have the forms
\[
\binom{\mathbf s_{\mathrm{term}}}{1}
=\mathsf A_{P;i,j}\binom{\mathbf s_{\mathrm{root}}}{1},
\qquad
\mathbf L_{\rm root}=\mathsf U_{P;i,j}\mathbf L_{\mathrm{term}}
              +\Pi_{P;i,j}(q,h),
\]
where $\mathsf U_{P;i,j}$ is induced by one affine coordinate map and one sign,
and every component of $\Pi_{P;i,j}$ is a bivariate polynomial of total
degree at most four.
Both composition and application cost $O(1)$ arithmetic operations.
\end{lemma}
Closure follows from affine substitution in the six degree-two monomials and
in the fixed degree-four boundary space.  The full proof is in
Appendix~\ref{app:operator-representation}.

\subsection{Compiling the cell operators}

Recall that $L$ is the number of one-dimensional true marker slots for the
current coefficient path.  The cell operator depends on $P$ and the two slot
indices, not on the numerical terminal pair $(M,R)$; all stopped pairs on the
path reuse it.  Pairing the $L$ slots of $u$ with the $L$ slots of $v$ gives
exactly $L^2$ true rectangles and hence $L^2$ cell operators.  Compile them
by sharing their common trace prefixes.  At each Euclidean step, marker slots
with the same trace so far form one group.  For every pair of groups, extend
the current operator once by their next quotient triple $(A,U,V)$ from
Definition~\ref{def:boundary-markers}.  When the traces finish, every pair
$(I_i,I_j)$ has its operator $\mathsf{Op}_{P;i,j}$.

\begin{lemma}[Shared-prefix compilation bound]
\label{lem:tree}
For a coefficient path with $L$ true marker slots, all $L^2$ cell operators
can be compiled in $O(L^2)$ arithmetic operations and stored in $O(L^2)$
arithmetic words.
\end{lemma}
The proof, based on the growth of the continuants along the prefix tree, is in
Appendix~\ref{app:operator-proofs}.

For the analysis, attach the constant-size correction menu of
Theorem~\ref{thm:virtual-rectangle-menu} to every uniform rectangle.  Its
entries identify the required cell operator and affine terminal data.  In the
reference implementation, equations~\eqref{eq:virtual-local-rank} and
\eqref{eq:virtual-code-to-slot} compute the two slot indices on demand, after
which the corresponding member of the $L^2$ operator array is accessed directly.
It suffices to retain these operators and the grid prefix data while processing
the stopped pairs of the current path and to discard them afterward; together
they require $O(\tau^2)$ path-local words.

\subsection{Small-modulus terminal lookup}
\label{sec:periodic-table}

The compiled operator removes the long common quotient trace but leaves a
small terminal staircase.  If its width or height is zero, or if $R=0$, all
six moments vanish.  If $M\leqslant\kappa$, the terminal query is read from the table
directly.  Otherwise Lemma~\ref{lem:onestep} permits one additional reciprocal
cycle, whose child coefficients are $(R,M\bmod R)$ and are both below
$\kappa$; after the child lookup, identity~\eqref{eq:one-step-lattice} lifts
its moments back to the original terminal query.

The lookup table stores one normalized staircase for every coprime pair of
coefficients at most $\kappa$.  The stored staircase has zero intercept, and we
keep prefixes of the six floor-power moments in~\eqref{eq:sixmom} over two
consecutive periods.  Equation~\eqref{eq:floor-lattice-conversion} then gives
the six terminal lattice moments required by the compiled operator.

A query whose coefficients are at most $\kappa$ is reduced to this stored
form in three constant-time steps.  First divide its two coefficients by their
greatest common divisor.  Second, use one precomputed modular inverse to turn
the intercept into a cyclic shift of the zero-intercept staircase; this shift
changes only the origin and adds a constant to the floor values.  Third, split
the requested length into complete periods and one final fragment.  The
fragment uses at most two stored prefixes, while all complete periods are
combined by ordinary power sums.

Horizontal and vertical shifts act affinely on the same six-moment family, so
the shifted prefixes are converted back by one fixed constant-size
calculation.
Thus the table contains no separate copy for every intercept: all intercepts
of one reduced coefficient pair reuse the same zero-intercept data.  The exact
shift identities are proved in Appendix~\ref{app:periodic-table}.

\begin{lemma}[Periodic table size]\label{lem:periodic}
The reduced periodic table, including the divisor lookup described in the
appendix, has
$O(\kappa^3)$ entries, is constructed in $O(\kappa^3)$ arithmetic operations,
and answers every terminal query in $O(1)$ operations.
\end{lemma}
The construction and size count are also proved in
Appendix~\ref{app:periodic-table}.

\subsection{Evaluation of one record}

This subsection applies to a long record, so its root length satisfies
$q>5\tau$.  Let
$\mathbf s_{\mathrm{root}}=(a,b;q,h;u,v)^{\mathsf T}$ be its affine query
state from Definition~\ref{def:affine-query-state}.  Select the first uniform
grid when $2R\leqslant M$ and the second one otherwise, and compute
$\operatorname{code}_\nu(u)=(x,\alpha_u)$ and
$\operatorname{code}_\nu(v)=(y,\alpha_v)$.
The menu entry $\mathcal M_{\nu,P}(x,y)[\alpha_u,\alpha_v]$ returns the exact true
indices $(i,j)$ and the operator $\mathsf{Op}_{P;i,j}$.  Write the terminal
state as $\mathbf s_{\mathrm{term}}=(M,R;q_{\mathrm{term}},h_{\mathrm{term}};
u_{\mathrm{term}},v_{\mathrm{term}})^{\mathsf T}$.
The affine component of the stored operator supplies it in homogeneous
coordinates as
$\binom{\mathbf s_{\mathrm{term}}}{1}
=\mathsf A_{P;i,j}\binom{\mathbf s_{\mathrm{root}}}{1}$.
If $R=0$, $q_{\mathrm{term}}=0$, or $h_{\mathrm{term}}=0$, the terminal
moments vanish.  Otherwise
Lemma~\ref{lem:onestep} permits at most one additional reciprocal step, and
Lemma~\ref{lem:periodic} returns in $O(1)$ operations the six lattice moments
$\mathbf L_{\mathrm{term}}$ of the original terminal query, including the
lift through this optional step.  Finally,
$\mathbf L_{\rm root}=\mathsf U_{P;i,j}\mathbf L_{\mathrm{term}}
+\Pi_{P;i,j}(q,h)$.
First convert $\mathbf L_{\rm root}$ to the six required $\Phi$-moments by
\eqref{eq:floor-lattice-conversion}.  Then
Equations~\eqref{eq:upper-moment-map}--\eqref{eq:moment-to-abc} add the record
to the layer accumulator.  Thus a long record is processed and accumulated
immediately during the visit to its coefficient path.
\section{Short-query fallback}\label{sec:early-termination}

Lemma~\ref{lem:momentreduction} emits at most one recursive floor query for
each coefficient pair.  In the normalized notation of~\eqref{eq:floor-query},
it is $\floor{(bt+\beta)/a}$ for $0\leqslant t<q$, where
$\beta\equiv N-bq\pmod a$ and $0\leqslant\beta<a$.
All other terms of the moment reduction are explicit power-sum expressions.
Call a nonempty query \emph{short} when $0<q\leqslant5\tau$, with $\tau$ from
\eqref{eq:layer-scales}.  Every short query is sent directly to the pointwise
recursion; every remaining nonempty query follows its compiled coefficient
path to the terminal table.

A short query is evaluated without a precompiled operator by the direct
six-floor-moment recurrence
\eqref{eq:floor-kernel-affine}--\eqref{eq:floor-kernel-reciprocal}.  It returns
$\boldsymbol\Phi$ itself, so no intermediate lattice-moment conversion is
needed.  The base cases $q=0$, $b=0$, and zero floor height return six zeros.
Every other call first extracts the integral coefficient and intercept parts;
the normalized reciprocal branch then exchanges the two positive
coefficients.  Thus the modulus sequence is the ordinary Euclidean sequence
and the recursion terminates.  The executable formulas and exact recursive
argument order are given in Appendix~\ref{app:operator-representation}.

\begin{theorem}[Short-query bound]\label{thm:early}
In a layer of size $N$, with the layer scale $D$ and path threshold $\tau$
from~\eqref{eq:layer-scales}, there are
$O(D\tau)=O(N^{3/4})$ short queries, and
their total evaluation time is
\[
 T_{\rm short}=O(D\tau\log(2+\tau))=O(N^{3/4}\log N)=o(N).
\]
Moreover, every query whose floor height would become zero before the end of
its compiled path is short.  Hence every nonempty non-short query can safely
follow the entire compiled path.
\end{theorem}
For a query that would terminate along its compiled path, reversing the length
recurrence bounds both root lengths by $5\tau$.  Counting all roots with
$0<q\leqslant5\tau$ gives $O(D\tau)$ queries, and their pointwise-recursion depth
is $O(\log D)=O(\log(2+\tau))$.  The full continuant calculation is in
Appendix~\ref{app:early-queries}.
\section{Assembling the complete \texorpdfstring{$O(n\log n)$}{O(n log n)} algorithm}
\label{sec:algorithmic-form}

The preceding sections provide all reusable components: the geometric
reduction creates one floor-moment record for each direction pair, the
coefficient path and uniform-grid correction choose its compiled operator, and
the small-modulus table supplies the terminal moments.  Their assembly has
two levels.

Inside one divisor layer, enumerate every direction pair and immediately add
the part of its contribution given by explicit power-sum expressions.  A nonempty
staircase also emits one record.  Records with the same coefficient path form
one batch.  For that batch construct the true marker slots, compile the cell
operators through their shared prefix tree, and prepare the two uniform grids
and their prefix data.  Each long record then selects one exact cell operator, and
its small terminal staircase is answered by the periodic table.  Short
records are handled by the pointwise recurrence.  The returned six moments
supply the deferred part of the geometric contribution.  After all records
of the path have been processed, discard its operators and grid data.

Outside the layer routine, quotient blocking groups all divisors that produce
the same layer size.  Three weighted M\"obius prefixes provide the weights of
each block, and one read-only periodic table is reused by every layer.  Thus
each distinct layer is constructed once and its three coefficients are added
to the final answer with the corresponding block weights.

The exact two-stage contribution formulas, the weighted M\"obius recurrence,
the cone-visitor pseudocode, a worked end-to-end record, and the formal
correctness chain are given in Appendix~\ref{app:complete-algorithm}.  All expensive objects are shared at the
path or global level, so processing one long record after its path is known
takes constant time.

\section{Linear layer and main theorem}\label{sec:linear-layer}

\begin{theorem}[Linear divisor layer]\label{thm:layer}
The coefficient triple $(A(N),B(N),C(N))$ can be computed exactly in
$O(N)$ arithmetic operations and $O(N^{3/4})$ arithmetic words of working
memory.
\end{theorem}
With $D=\Theta(N^{1/2})$ and
$\tau,\kappa=\Theta(N^{1/4})$, record generation and application cost
$O(D^2)$, constructing all cell operators and correction data costs $O(\tau^4)$, the
terminal table costs $O(\kappa^3)$, and short queries are lower order.  Hence the time is $O(N)$;
the terminal table dominates memory at $O(N^{3/4})$.  The complete cost table
and the reason the three scales cannot all be decreased are in
Appendix~\ref{app:one-value-complexity}.

\begin{lemma}[Distinct-layer sum]\label{lem:distinctlayers}
Let $\mathcal V_n=\{\floor{n/d}:1\leqslant d\leqslant n\}$ be the layer
sizes evaluated after quotient blocking.  Then
$\sum_{N\in\mathcal V_n}N=O(n\log n)$.
\end{lemma}
\begin{proof}
Put $D_0=\floor{\sqrt n}$.  Quotients at most $D_0$ contribute
$O(D_0^2)=O(n)$.  Every larger distinct quotient has the form
$\floor{n/d}$ with $d\leqslant D_0$, and
$\sum_{d\leqslant D_0}\floor{n/d}=O(n\log n)$.
\end{proof}

\begin{theorem}[Main theorem]\label{thm:main}
The exact number $F(n)$ of lattice rectangles in an $n\times n$ square grid
can be computed in $O(n\log n)$ arithmetic operations and
$O(n^{3/4})$ arithmetic words of working memory.
\end{theorem}
\begin{proof}
Apply Theorem~\ref{thm:layer} inside the divisor-layer formula
\eqref{eq:layers}.  Quotient blocking evaluates each
$N\in\mathcal V_n$ once, and Lemma~\ref{lem:distinctlayers} gives
$O(n\log n)$ total layer work.

The three weighted M\"obius prefixes in~\eqref{eq:mobius-prefix} take
$O(n^{2/3})$ time and storage by a cutoff sieve followed by quotient-block
recurrences; the calculation is given in
Appendix~\ref{app:mobius-prefix-memory}.  This is lower order, while sequential
layer processing uses $O(n^{3/4})$ peak memory.
\end{proof}

\begin{lemma}[One-value operand size]\label{lem:one-value-operands}
Every exact integer stored or multiplied by the one-value algorithm has
$O(\log n)$ bits.
\end{lemma}
All coordinates are polynomially bounded in $n$, and composition has fixed
dimension and logarithmic depth.  The coefficient-growth calculation is in
Appendix~\ref{app:one-value-complexity}.

\section{Experiments}\label{sec:one-value-experiments}

We compare the $O(n\log n)$ algorithm with the exact
$O(n\log^2n)$ divisor-layer implementation from the public code repository
on an AMD Ryzen 9 5950X running Windows.  Both implementations use the same
experimental interface and return exactly one prescribed value $F(n)$ per
invocation.  GCC 8.1.0 from MinGW-w64 compiled both programs with
\texttt{-O3 -DNDEBUG -march=native -flto}.

At each $n=2^k$, both programs are run ten times.  Every timed computation starts with an
empty auxiliary cache and includes all preprocessing performed by the
corresponding implementation; only output is excluded.  Both processes are
pinned to a fixed logical processor, run at high
priority under the Windows high-performance power plan.

Figure~\ref{fig:one-value-timing} reports the arithmetic means of the ten
elapsed times at $n=2^k$.  In the right panel, each running time is normalized
by the corresponding theoretical complexity; values are in seconds.

\begin{figure}[H]
  \centering
  \resizebox{\textwidth}{!}{%
  \begin{tikzpicture}
  \begin{groupplot}[
    group style={group size=2 by 1,horizontal sep=1.75cm},
    width=6.05cm,height=4.00cm,
    xmin=9.5,xmax=25.5,
    xtick={10,11,12,13,14,15,16,17,18,19,20,21,22,23,24,25},
    xticklabels={$2^{10}$,$2^{11}$,$2^{12}$,$2^{13}$,$2^{14}$,
                 $2^{15}$,$2^{16}$,$2^{17}$,$2^{18}$,$2^{19}$,
                 $2^{20}$,$2^{21}$,$2^{22}$,$2^{23}$,$2^{24}$,$2^{25}$},
    xlabel={grid size $n=2^k$},
    grid=major,grid style={benchmarkgrid!65},
    axis line style={black!80},
    x tick label style={font=\scriptsize,rotate=45,anchor=east},
    y tick label style={font=\scriptsize},
    label style={font=\footnotesize},
    xlabel style={yshift=2pt},
    title style={font=\footnotesize,
                 at={(axis description cs:0.5,1.005)},anchor=south},
    legend style={font=\scriptsize,draw=benchmarkgrid,fill=white,
                  fill opacity=0.94,text opacity=1,cells={anchor=west},
                  inner xsep=3pt,inner ysep=2pt}
  ]
    \nextgroupplot[
      title={Running times of the algorithms},
      ylabel={running time (seconds)},
      ymode=log,ymin=0.0003,ymax=100,
      ytick={0.001,0.01,0.1,1,10,100},
      legend pos=north west
    ]
      \addplot[benchmarkbaseline,thick,mark=triangle*,mark size=2.3pt]
        coordinates {
          (10,0.000532890)(11,0.001189880)(12,0.002717190)
          (13,0.005984020)(14,0.014042840)(15,0.031578500)
          (16,0.068049890)(17,0.147840260)(18,0.332526450)
          (19,0.734904960)(20,1.631658100)(21,3.354784790)
          (22,7.269896950)(23,15.796421190)(24,34.597672780)
          (25,73.533618160)
        };
      \addlegendentry{$O(n\log_2^2 n)$}
      \addplot[benchmarkproposed,thick,mark=square*,mark size=2.1pt]
        coordinates {
          (10,0.001447670)(11,0.001744720)(12,0.002785460)
          (13,0.006319030)(14,0.014831100)(15,0.025661360)
          (16,0.042857430)(17,0.104808200)(18,0.200522790)
          (19,0.446458970)(20,0.868196060)(21,1.837163590)
          (22,3.843660720)(23,8.341523900)(24,16.931492340)
          (25,36.384306720)
        };
      \addlegendentry{$O(n\log_2 n)$}

    \nextgroupplot[
      title={Normalized running times},
      ylabel={normalized running time (ns)},
      ymode=log,ymin=3,ymax=170,
      ytick={3,10,30,100},
      yticklabels={$3$,$10$,$30$,$100$},
      legend style={
        at={(axis description cs:0.98,0.38)},
        anchor=east
      }
    ]
      \addplot[benchmarkbaseline,thick,mark=triangle*,mark size=2.3pt]
        coordinates {
          (10,5.204003906)(11,4.801620610)(12,4.606781006)
          (13,4.322314742)(14,4.372994559)(15,4.283108181)
          (16,4.056089520)(17,3.902877755)(18,3.915085616)
          (19,3.882880805)(20,3.890176058)(21,3.627405999)
          (22,3.581153928)(23,3.559697821)(24,3.580177077)
          (25,3.506356152)
        };
      \addlegendentry{$O(n\log_2^2 n)$}
      \addplot[benchmarkproposed,thick,mark=square*,mark size=2.1pt]
        coordinates {
          (10,141.374023)(11,77.446733)(12,56.670329)
          (13,59.335844)(14,64.658465)(15,52.208171)
          (16,40.872030)(17,47.036653)(18,42.496319)
          (19,44.818574)(20,41.398814)(21,41.715613)
          (22,41.654551)(23,43.234223)(24,42.049816)
          (25,43.373474)
        };
      \addlegendentry{$O(n\log_2 n)$}
  \end{groupplot}
  \end{tikzpicture}%
  }
  \caption{Mean one-value running times over ten runs on logarithmic scales.
  Left: elapsed time.  Right: time normalized by the corresponding theoretical
  complexity.}
  \label{fig:one-value-timing}
\end{figure}
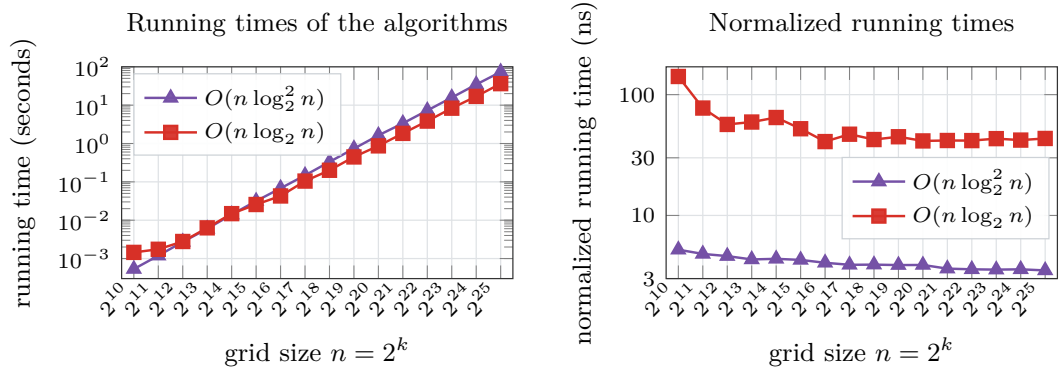

On larger inputs, the broadly stable normalized curves support the predicted
scaling.  Small residual high-end fluctuations may reflect growing working
sets crossing cache-capacity thresholds.  Despite its lower $O(n\log n)$
complexity, the new algorithm has a larger constant from storing and applying
relatively heavy compiled kernels; its lower asymptotic growth compensates on
sufficiently large inputs.

\section{Reproducibility}

The reference implementations and build instructions are available at
\url{https://github.com/flykiller/lattice-rectangles-one-value}.

\section{Future work}\label{sec:future-work}

The reference C++ implementation is single-threaded.  Divisor layers and
coefficient paths are independent, but a parallel version should share the
M\"obius prefixes, marker data, and terminal table among workers.

The terminal periodic table causes the $O(n^{3/4})$ working-memory term; the
live data for one coefficient path is smaller.  Achieving the same
$O(n\log n)$ arithmetic bound with a smaller terminal representation remains
open, perhaps by evaluating selected periodic prefixes on demand while
preserving enough sharing to avoid restoring a logarithmic factor in time.

\section{Conclusion}

For a prescribed $n\times n$ grid, the exact lattice-rectangle count is
computable in $O(n\log n)$ arithmetic operations and $O(n^{3/4})$ arithmetic
words of working memory.  The method compiles shared prefixes of weighted
Euclidean floor-moment recurrences, uses uniform marker grids to locate
queries, and resolves constant-size local boundary corrections.  Its operands
have $O(\log n)$ bits, and the implementation exhibits the predicted
$n\log n$ scaling over the tested range.

\clearpage
\appendix

\section{Geometric and moment identities}
\label{app:part1-moments}

\subsection{Full half-domain moment reduction}
\label{app:half-domain-moments}

This section supplies the explicit formulas used by
Lemma~\ref{lem:momentreduction}.  For the strict non-diagonal contribution,
the direction pair $(a,b)$ and the side pair $(x,y)$ both have decreasing
positive coordinates.  The
constraint $ax+by\leqslant N$ and the vector
$\mathbf w_{a,b}(x,y)$ are unchanged when the two pairs are exchanged, and
the multiplier is four.  Pairing $a<x$ with $a>x$ gives
\begin{equation}\label{eq:halfdomain}
4\!\sum_{\substack{a>b\geqslant1,\ x>y\geqslant1\\ax+by\leqslant N}}
\mathbf w_{a,b}(x,y)
=
8\!\sum_{\substack{a>b\geqslant1,\ x>y\geqslant1\\a\leqslant x,\ ax+by\leqslant N}}
\mathbf w_{a,b}(x,y)
-
4\!\sum_{\substack{a>b\geqslant1,\ x>y\geqslant1\\a=x,\ ax+by\leqslant N}}
\mathbf w_{a,b}(x,y).
\end{equation}
In the first sum on the right, $ax\leqslant N$ and $a\leqslant x$, hence
$a\leqslant D=\floor{\sqrt N}$.

Fix $(a,b)$ with $a\leqslant D$.  For a fixed $y$, the admissible $x$ form
the interval
\[
x_{\min}(y)\leqslant x\leqslant x_{\max}(y),\qquad
x_{\max}(y)=\floor{\frac{N-by}{a}},\qquad
x_{\min}(y)=
\begin{cases}
a,&y<a,\\
y+1,&y\geqslant a.
\end{cases}
\]
Put
\[
q_{\mathrm{sq}}=\floor{\frac{N-a^2}{b}},\qquad
q=
\begin{cases}
q_{\mathrm{sq}},&q_{\mathrm{sq}}<a-1,\\
\displaystyle\max\!\left\{a-1,\floor{\frac{N-a}{a+b}}\right\},
&q_{\mathrm{sq}}\geqslant a-1.
\end{cases}
\]
For $y<a$, nonemptiness is $y\leqslant q_{\mathrm{sq}}$; if
$q_{\mathrm{sq}}\geqslant a-1$, the
remaining condition is $(a+b)y+a\leqslant N$.  Hence the admissible $y$
form exactly $1\leqslant y\leqslant q$.

For $s\geqslant0$ define
$\Sigma_s(z)=\sum_{1\leqslant k\leqslant z}k^s$, with value zero for
$z\leqslant0$.  For $i+j\leqslant2$ the required half-domain moments are
\begin{equation}\label{eq:halfmoment}
J_{ij}(a,b;N)=
\sum_{y=1}^{q}y^j
\left[
\Sigma_i\!\left(\floor{\frac{N-by}{a}}\right)
-\Sigma_i(x_{\min}(y)-1)
\right].
\end{equation}
Write $N-bq=\eta a+\beta$, $0\leqslant\beta<a$, and form the single query
\begin{equation}\label{eq:pair-floor-query}
f(t)=\floor{\frac{bt+\beta}{a}},\qquad 0\leqslant t<q.
\end{equation}
Reversing $y=q-t$ gives $x_{\max}(q-t)=\eta+f(t)$.  Write
$\Phi_{ij}=\Phi_{ij}(f)$ for the six moments in~\eqref{eq:sixmom}.  Then the
upper-limit sums are
\begin{align}
J^{\mathrm{upper}}_{00}&=q\eta+\Phi_{01},\notag\\
J^{\mathrm{upper}}_{01}&=\eta\Sigma_1(q)+q\Phi_{01}-\Phi_{11},\notag\\
J^{\mathrm{upper}}_{02}&=\eta\Sigma_2(q)+q^2\Phi_{01}-2q\Phi_{11}+\Phi_{21},\notag\\
J^{\mathrm{upper}}_{10}&=\frac{q(\eta^2+\eta)+(2\eta+1)\Phi_{01}+\Phi_{02}}2,\notag\\
J^{\mathrm{upper}}_{11}&=\frac{(\eta^2+\eta)\Sigma_1(q)
 +(2\eta+1)(q\Phi_{01}-\Phi_{11})
 +q\Phi_{02}-\Phi_{12}}2,\notag\\
J^{\mathrm{upper}}_{20}&=\frac{q(2\eta^3+3\eta^2+\eta)
 +(6\eta^2+6\eta+1)\Phi_{01}
 +(6\eta+3)\Phi_{02}+2\Phi_{03}}6.       \label{eq:upper-moment-map}
\end{align}
The lower-limit sums are
\begin{equation}\label{eq:lower-moment-map}
J^{\mathrm{lower}}_{ij}=\Sigma_i(a-1)\Sigma_j(\min\{q,a-1\})
       +\sum_{y=a}^{q}y^j\Sigma_i(y),
\qquad J_{ij}=J^{\mathrm{upper}}_{ij}-J^{\mathrm{lower}}_{ij}.
\end{equation}
The last sum is empty when $q<a$.  Since $i+j\leqslant2$,
$y^j\Sigma_i(y)$ is a polynomial of degree at most three and is evaluated by
$\Sigma_0,\ldots,\Sigma_3$ in $O(1)$ operations.

For this fixed pair, put
\begin{equation}\label{eq:moment-linear-map}
\mathcal A_{a,b}(J)=\left(
J_{00},
(a+b)(J_{10}+J_{01}),
ab(J_{20}+J_{02})+(a^2+b^2)J_{11}
\right).
\end{equation}
The first sum on the right of~\eqref{eq:halfdomain} contributes
\begin{equation}\label{eq:moment-to-abc}
8\mathcal A_{a,b}(J).
\end{equation}
The equality boundary $a=x$ has
$1\leqslant y\leqslant y_\star$, where
$y_\star=\min\{a-1,q_{\mathrm{sq}}\}$.  Its subtracted vector is
\begin{equation}\label{eq:equality-correction}
E_{a,b}=4\left(
 y_\star,
 (a+b)(y_\star a+\Sigma_1(y_\star)),
 ab(y_\star a^2+\Sigma_2(y_\star))
 +(a^2+b^2)a\Sigma_1(y_\star)
\right).
\end{equation}

It remains to add the side diagonal $x=y=s$, which has multiplicity two.
For $c=a+b$, the number of pairs $a>b\geqslant1$ with sum $c$ is
$\floor{(c-1)/2}$.  Put
\[
\operatorname{Diag}_t(m)=
\sum_{c=3}^{m}c^t\floor{\frac{c-1}{2}},\qquad t=0,1,2.
\]
For $k_+=\floor{m/2}$ and $k_-=\floor{(m-1)/2}$, separating even and odd
$c$ gives
\begin{align}
\operatorname{Diag}_0(m)&=\Sigma_1(k_+)-k_++\Sigma_1(k_-),\notag\\
\operatorname{Diag}_1(m)&=2(\Sigma_2(k_+)-\Sigma_1(k_+))
 +2\Sigma_2(k_-)+\Sigma_1(k_-),\notag\\
\operatorname{Diag}_2(m)&=4(\Sigma_3(k_+)-\Sigma_2(k_+))+4\Sigma_3(k_-)
       +4\Sigma_2(k_-)+\Sigma_1(k_-).             \label{eq:diagonal-prefix}
\end{align}
Therefore the complete diagonal contribution is
\begin{equation}\label{eq:diagonal-contribution}
\mathcal D(N)=
\sum_{s=1}^{\floor{N/3}}
\left(2\operatorname{Diag}_0(\floor{N/s}),
      4s\operatorname{Diag}_1(\floor{N/s}),
      2s^2\operatorname{Diag}_2(\floor{N/s})\right).
\end{equation}
It takes $O(N)$ operations and requires no floor query.  Combining all pieces,
\begin{equation}\label{eq:complete-layer-generation}
(A(N),B(N),C(N))
=\mathcal D(N)+
\sum_{1\leqslant b<a\leqslant D}
\left(8\mathcal A_{a,b}(J)-E_{a,b}\right).
\end{equation}
This proves Lemma~\ref{lem:momentreduction}, including its generation bound.
\section{Technical proofs for Euclidean batching}
\label{app:part1-batching}

\subsection{Derivation of one complete Euclidean cycle}
\label{app:cycle-derivation}

The reciprocal query~\eqref{eq:reciprocal-query} is not yet normalized because
its coefficient $a$ can exceed its modulus $b$.  Divide the three quantities
that control the next step by $b$:
\begin{equation}\label{eq:stepquotients}
a=Ab+b',\qquad u=Ub+\beta',\qquad v=Vb+\gamma,
\end{equation}
where
\[
A=\floor{a/b},\qquad U=\floor{u/b},\qquad
V=\floor{v/b},\qquad 0\leqslant b',\beta',\gamma<b.
\]
Then
\begin{equation}\label{eq:child-query}
\widehat f(k)=Ak+U+f'(k),\qquad
f'(k)=\floor{\frac{b'k+\beta'}{b}},\qquad 0\leqslant k<h.
\end{equation}
The polynomial part $Ak+U$ is handled explicitly by the six-moment affine
formulas.  The only recursive child is $f'$, with
\[
a'=b,\qquad b'=a-Ab,\qquad
\beta'=u-Ub,\qquad q'=h.
\]
Renormalizing it gives
\[
u'=a'-\beta'-1=(U+1)b-u-1,\qquad
v'=b-\gamma-1=(V+1)b-v-1.
\]
Finally, substituting $b(q-1)+\beta=ah+v$ in the endpoint formula for $f'$ yields
\[
h'=q-Ah-(U+V+2-A).
\]
Together these identities give the cycle~\eqref{eq:cycle}.

\subsection{Coefficient cones and terminal coordinates}
\label{app:coefficient-cones}

\begin{proof}[Proof of Lemma~\ref{lem:onestep}]
The assertion is immediate if $R=0$.  Otherwise suppose $M>\kappa$ and let
$A=\floor{M/R}$.  If $(\pi_{11},\pi_{12})$ is the first row of $P$, rejection
of the next child gives
\[
 \pi_{11}A+\pi_{12}>\tau.
\]
For $P=I$ this is simply $A>\tau$.  For a nonempty continuant product, its
first row dominates the second row componentwise, so an overflowing entry in
$PQ(A)$ implies the same displayed inequality.  Therefore
\[
 a=\pi_{11}M+\pi_{12}R
 \geqslant(\pi_{11}A+\pi_{12})R>\tau R.
\]
Since $a\leqslant D\leqslant\tau\kappa$, it follows that $R<\kappa$.
The next Euclidean modulus is $R$.
\end{proof}

\begin{proof}[Proof of Lemma~\ref{lem:terminal-bijection}]
Starting from any root pair, follow its canonical Euclidean divisions.  Enter
a child whenever the current terminal pair has $M>\tau$ and $R>0$ and the
child matrix is accepted.  The first failed condition is uniquely one of the three stop
conditions in Section~\ref{sec:terminal-coordinates}.  Conversely each stopped
pair reverses through $P$ to one root pair.  Unimodularity makes $(M,R)$
unique, proving the bijection.

For the matrix count, an accepted nonidentity matrix has nonnegative entries,
determinant $\pm1$, and entries at most $\tau$.  Write its rows as
$(p_1,p_2)$ and $(r_1,r_2)$.  The first row is primitive, and the continuant
inequalities give
\begin{equation}\label{eq:continuant-dominance}
 0\leqslant r_1\leqslant p_1,\qquad
 0\leqslant r_2\leqslant p_2.
\end{equation}
Fix $(p_1,p_2)$ and a determinant sign $\delta$.  If
$(r_{1,0},r_{2,0})$ is one solution of
$p_1r_2-p_2r_1=\delta$, every solution is
\[
 (r_1,r_2)=(r_{1,0},r_{2,0})+t(p_1,p_2),
 \qquad t\in\mathbb Z.
\]
For a nonidentity path both $p_1$ and $p_2$ are positive.  Each inequality
in~\eqref{eq:continuant-dominance} therefore restricts $t$ to a closed
interval of length at most one, so their intersection contains at most two
integers.  There are $O(\tau^2)$ possible first rows and two determinant
signs.  A recursion word is also recovered uniquely from its matrix.  Indeed,
if $P=P^-Q(A)$ is nonidentity, then
\[
 A=\min\!\left\{\floor{p_1/p_2},\floor{r_1/r_2}\right\},
\]
where a quotient with zero denominator is omitted, and
$P^-$ has first column $(p_2,r_2)^{\mathsf T}$ and second column
$(p_1-Ap_2,r_1-Ar_2)^{\mathsf T}$.  Repeating this step reaches $I$.
Thus a matrix is not counted through two recursion words.  Including $I$
gives $O(\tau^2)$ nodes.

The small-modulus piece tests fewer than $\tau^2$ pairs per node and therefore
costs $O(\tau^4)$.  The zero-remainder piece is one integer interval in $M$.
For the overflow piece put
$A_{\rm stop}=\operatorname{ext}(P)+1$.  For every
\[
 1\leqslant R\leqslant
 \floor{\frac{D}{\pi_{11}A_{\rm stop}+\pi_{12}}},
\]
the admissible $M$ form the explicit interval
\[
 \max\{\tau+1,A_{\rm stop}R,R+1\}
 \leqslant M\leqslant
 \floor{\frac{D-\pi_{12}R}{\pi_{11}}},
\]
intersected with
$1\leqslant\pi_{21}M+\pi_{22}R
 <\pi_{11}M+\pi_{12}R$.
There is one constant-time interval test per candidate $R$.  Summed over all
nodes these tests cost $O(D\tau^2)$, which is
$O(D^2+\tau^4)$ by $2D\tau^2\leqslant D^2+\tau^4$.  The emitted pairs
themselves number exactly $D(D-1)/2$.
\end{proof}

\subsection{Boundary formulas and uniform-grid corrections}
\label{app:marker-orbit-proof}

Fix an accepted coefficient path
\[
 (a_0,b_0),(a_1,b_1),\ldots,(a_d,b_d)=(M,R),
 \qquad a_i=A_i b_i+r_i,
\]
where $(a_{i+1},b_{i+1})=(b_i,r_i)$.  Put $H=a_0$ and $B_0=b_0$.  A root
marker evolves according to
\begin{equation}\label{eq:fold}
 U_i=\floor{u_i/b_i},
 \qquad
 u_{i+1}=b_i-1-(u_i\bmod b_i).
\end{equation}

\begin{theorem}[Modular-orbit boundary formula]
\label{thm:marker-orbit}
Assume $R>0$.
For $1\leqslant m\leqslant d$ let
\[
 P_m=Q(A_0)\cdots Q(A_{m-1}),
 \qquad \chi_m=(P_m)_{11},
 \qquad L=1+\sum_{m=1}^{d}\chi_m.
\]
Then the boundary multiset of Definition~\ref{def:marker-intervals} is
\begin{equation}\label{eq:orbit-boundaries}
 (e_0,\ldots,e_{L-1})
 =\operatorname{sort}\bigl((kB_0\bmod H)_{0\leqslant k<L}\bigr),
 \qquad e_L=H.
\end{equation}
The listed residues are distinct.
\end{theorem}

\begin{proof}[Proof of Theorem~\ref{thm:marker-orbit}]
Put $P_0=I$ and write every prefix matrix as
\[
 P_i=
 \begin{pmatrix}\chi_i&\chi_{i-1}\\[1mm]\eta_i&\eta_{i-1}\end{pmatrix},
 \qquad \chi_{-1}=0,\quad\chi_0=1.
\]
Thus $\chi_{i+1}=A_i\chi_i+\chi_{i-1}$ and
$\det P_i=(-1)^i$.  Since
$(H,B_0)^{\mathsf T}=P_i(a_i,b_i)^{\mathsf T}$, we have
\begin{equation}\label{eq:orbit-return}
 \chi_iB_0-\eta_iH=(-1)^i b_i.
\end{equation}

Consider the edge between two consecutive root markers $u-1$ and $u$.  As
long as their earlier marker quotients agree, the fold~\eqref{eq:fold} sends
the edge coordinate $x_i$ to
\[
 x_{i+1}=(-x_i)\bmod b_i.
\]
The quotient $U_i$ changes across that edge exactly when $x_i=0$ modulo
$b_i$.  Let
$L_i=1+\sum_{m=1}^{i}\chi_m$.

\emph{Induction claim.}
After the first $i$ marker quotients have been fixed, their cut edges are
\[
 \{kB_0\bmod H:0\leqslant k<L_i\}.
\]
For $i=0$, $L_0=1$ and the set consists only of the edge zero.

For the induction step, fix an atom $J$ of the partition made by these edges.
Unrolling the first $i$ folds on $J$ writes
\[
 u_i=(-1)^iu+c_J
\]
for an integer constant $c_J$ depending on that atom.  Consequently $U_i$
is constant until this affine lift crosses a multiple of $b_i$; at such a
crossing it changes by one.  Hence locating the new cut edges is equivalent to
listing all crossings of consecutive $b_i$-blocks by these affine lifts.
Equation~\eqref{eq:orbit-return} makes the lift of a root
translation explicit:
\[
 (-1)^i\chi_iB_0\equiv b_i\pmod H.
\]
Thus translating a root edge by $\chi_iB_0$ advances its lifted value by one
whole $b_i$-block.  Index the new root edges from the end of the old orbit
block, so that their orbit indices are $L_i+j$.  The division
$a_i=A_ib_i+b_{i+1}$ supplies $A_i$ complete blocks.  Each complete block has
$\chi_i$ lifted positions, and the wraparound recorded by the second column of
$P_i$ has $\chi_{i-1}$ positions.  Equivalently, the possible offsets $j$
split without overlap as
\[
 \{r\chi_i+\ell:0\leqslant r<A_i,
                    0\leqslant\ell<\chi_i\}
 \;\dot\cup\;
 \{A_i\chi_i+\ell:0\leqslant\ell<\chi_{i-1}\}.
\]
For fixed $r$, the first set is the complete offset block
$[r\chi_i,(r+1)\chi_i)$ and records the next crossing in each lifted
position.  The second set is exactly the incomplete wrapped block.  Thus every
displayed offset gives a crossing, and any other edge stays strictly inside
one $b_i$-block.  The two sets form, without gaps or overlaps, the interval
$0\leqslant j< A_i\chi_i+\chi_{i-1}
 =\chi_{i+1}$.  Here the last equality is the continuant recurrence.  Hence
the new crossing edges are exactly
\[
 \{(L_i+j)B_0\bmod H:0\leqslant j<\chi_{i+1}\}.
\]
They extend the preceding indices to $0\leqslant k<L_{i+1}$; the rotation
period bound verified below makes all corresponding root edges disjoint.  This
proves the claim by induction and, at $i=d$, proves
\eqref{eq:orbit-boundaries} with $L=L_d$.

The same recurrence gives
\begin{equation}\label{eq:orbit-count-bound}
 L\leqslant2\chi_d+\chi_{d-1}.
\end{equation}
For the induction step, add $\chi_{i+1}$ to
$L_i\leqslant2\chi_i+\chi_{i-1}$ and use
$\chi_{i+1}\geqslant\chi_i+\chi_{i-1}$.

The orbit points are exactly the cut edges of the marker-word map, so the word
is fixed precisely between consecutive distinct boundaries.  To verify the
distinctness statement, let $g=\gcd(M,R)$.  Unimodularity gives
$\gcd(H,B_0)=g$.  Since $R>0$, we have $M/g\geqslant2$ and $R/g\geqslant1$, and hence
the rotation period
$H/g=\chi_d(M/g)+\chi_{d-1}(R/g)$ is at least
$2\chi_d+\chi_{d-1}\geqslant L$.  Thus the first $L$
residues are distinct.  This completes the proof of the theorem.
\end{proof}

\pagebreak
\medskip
\noindent\textbf{Affine endpoints.}

Write
$P=\left(\begin{smallmatrix}\pi_{11}&\pi_{12}\\
\pi_{21}&\pi_{22}\end{smallmatrix}\right)$ and use the
interior sample pair $(M,R)=(2,1)$.  Put
\[
 H_*=2\pi_{11}+\pi_{12},
 \qquad B_*=2\pi_{21}+\pi_{22}.
\]
Let $k_0,\ldots,k_{L-1}$ order the residues $kB_*\bmod H_*$ increasingly and
put $c_k=\floor{kB_*/H_*}$.  Define
\begin{equation}\label{eq:affine-endpoint-coefficients}
 \lambda_j=k_j\pi_{21}-c_{k_j}\pi_{11},
 \qquad
 \omega_j=k_j\pi_{22}-c_{k_j}\pi_{12}.
\end{equation}

To prove that this sample order is valid throughout the cone, put $\xi=R/M$.
The rotation number is
\[
 \frac {B_0}H=\frac{\pi_{21}+\pi_{22}\xi}{\pi_{11}+\pi_{12}\xi},
\]
which lies between $\pi_{21}/\pi_{11}$ and
$(\pi_{21}+\pi_{22})/(\pi_{11}+\pi_{12})$.  These two fractions are Farey
neighbours because
$\pi_{11}\pi_{22}-\pi_{12}\pi_{21}=\det P=\pm1$.  A floor
$\floor{kB_0/H}$ with $k<L$ can change only when $B_0/H$ reaches a rational
whose reduced denominator divides $k$; the order of residues with indices
$k,\ell<L$ can change only when $(k-\ell)B_0/H$ is integral.  Either event
therefore has a reduced denominator below $L$.  By
\eqref{eq:orbit-count-bound}, $L\leqslant2\pi_{11}+\pi_{12}$, whereas no
rational with denominator below this value lies strictly between Farey
neighbours whose denominator sum is $2\pi_{11}+\pi_{12}$.  Hence all floors
and the residue order are constant for $0<R<M$ and may be computed at the
sample point $\xi=1/2$.  Expanding
\[
 k_jB_0-c_{k_j}H
\]
in $M$ and $R$ and using~\eqref{eq:affine-endpoint-coefficients} gives
\begin{equation}\label{eq:endpoint}
 e_j=\lambda_jM+\omega_jR.
\end{equation}
For $R=0$ we define the labeled boundary list by the closed specialization of
these affine forms.  Since $e_j(M,R)\leqslant e_{j+1}(M,R)$ for every
$0<R<M$, passage to the limit gives a nondecreasing list at $R=0$; some
entries may coalesce at $0$ or $H$.  Each quotient test in~\eqref{eq:fold} is
piecewise affine between consecutive cut edges; at a cut edge the floor and
the half-open convention both select the slot on its right.  Passing these
inequalities to $R=0$ therefore preserves the labeled marker word on every
positive-width limiting slot, while intervals that collapse are empty.  Thus the specialized list,
with repetitions retained and $e_L=H$, gives exactly the true labeled slots at
$R=0$.  It is this affine specialization, rather than the ordinary modular
residue multiset, that is used for zero-remainder terminal pairs.

Once the marker word is known, no further recurrence is needed at query time.
Let $b_i=\vartheta_{i,1}M+\vartheta_{i,2}R$ be the second coefficient at
depth $i$ and let $U_i$ be the fixed marker quotient of one true slot.  Unrolling
\eqref{eq:fold} gives the executable terminal-marker formula
\begin{equation}\label{eq:terminal-marker-affine}
 u_d=(-1)^d u_0+
 \sum_{i=0}^{d-1}(-1)^{d-1-i}
 \bigl((U_i+1)(\vartheta_{i,1}M+\vartheta_{i,2}R)-1\bigr).
\end{equation}
The same formula with the second marker word $(V_i)_i$ gives $v_d$.
Consequently each slot stores four integral coefficients for
$u_d=c_u u_0+c_MM+c_RR+c_0$; a slot pair and its moment operator determine
the complete terminal state in constant time.

\medskip
\noindent\textbf{Uniform grids and local corrections.}
Put $\delta=\det P$.  The coefficients in~\eqref{eq:endpoint} satisfy
\[
 \pi_{11}\omega_j-\pi_{12}\lambda_j=k_j\delta.
\]
Consequently the two endpoint specializations give the exact identities
\begin{equation}\label{eq:virtual-endpoint-identities}
 \pi_{11} e_j=\lambda_jH+k_j\delta R,
 \qquad
 (\pi_{11}+\pi_{12})e_j=(\lambda_j+\omega_j)H-k_j\delta(M-R).
\end{equation}
For the technical construction define
\[
 T_0=\pi_{11},\quad z_{0,j}=\lambda_j,
 \qquad
 T_1=\pi_{11}+\pi_{12},\quad z_{1,j}=\lambda_j+\omega_j,
\]
and let $K_\nu(z)$ be the multiset of indices $k_j$ with
$z_{\nu,j}=z$.  For a terminal pair select
\begin{equation}\label{eq:virtual-grid-selection}
 (\nu,T,E,\varepsilon)=
 \begin{cases}
  (0,\pi_{11},R,\det P),&2R\leqslant M,\\
  (1,\pi_{11}+\pi_{12},M-R,-\det P),&2R>M.
 \end{cases}
\end{equation}
The selected uniform grid divides $[0,H)$ into $T$ equal intervals; the cell
of a marker $t$ is $x=\floor{Tt/H}$.

At the two endpoint specializations of~\eqref{eq:endpoint},
$0\leqslant e_j\leqslant H$ gives
\[
 0\leqslant\lambda_j\leqslant\pi_{11},\qquad
 0\leqslant\lambda_j+\omega_j\leqslant\pi_{11}+\pi_{12}.
\]
Thus $K_0$ and $K_1$ are stored directly in arrays indexed respectively by
$0,\ldots,T_0$ and $0,\ldots,T_1$; no associative search structure is needed.

For $2R\leqslant M$, the orbit bound~\eqref{eq:orbit-count-bound} gives
\[
 k_jR<(2\pi_{11}+\pi_{12})R
 \leqslant \pi_{11}M+\pi_{12}R=H.
\]
For $2R>M$ it gives
\[
 k_j(M-R)<(2\pi_{11}+\pi_{12})(M-R)
 <\pi_{11}M+\pi_{12}R=H.
\]
Thus the displacement term in the selected identity of
\eqref{eq:virtual-endpoint-identities} is strictly smaller than one uniform
cell.

It remains to bound the number of orbit points at one grid position.
The continuant recurrence gives $0\leqslant\pi_{12}\leqslant\pi_{11}$.
If $\lambda_j=\lambda_{j'}$, then
$(k_j-k_{j'})\pi_{21}$ is divisible by $\pi_{11}$.  Unimodularity gives
$\gcd(\pi_{11},\pi_{21})=1$, so the corresponding orbit indices are
congruent modulo $\pi_{11}$.
Since $0\leqslant k_j<L\leqslant2\pi_{11}+\pi_{12}
\leqslant3\pi_{11}$, at most three indices share
one position of the first grid.  Similarly,
$\gcd(\pi_{11}+\pi_{12},\pi_{21}+\pi_{22})=1$; equality of
$\lambda_j+\omega_j$ and $\lambda_{j'}+\omega_{j'}$ makes the two indices
congruent modulo $\pi_{11}+\pi_{12}$.
The bound $L\leqslant2\pi_{11}+\pi_{12}
\leqslant2(\pi_{11}+\pi_{12})$ leaves at most two indices at one
position of the second grid.

For the selected tuple~\eqref{eq:virtual-grid-selection}, the corresponding
identity can be written
\begin{equation}\label{eq:selected-virtual-identity}
 T e_j=z_{\nu,j}H+\varepsilon k_jE,
 \qquad 0\leqslant k_jE<H.
\end{equation}
For later use define the inclusive group prefix
\[
 \operatorname{pref}_\nu(z)=\sum_{w=0}^{z}|K_\nu(w)|,
 \qquad \operatorname{pref}_\nu(-1)=0.
\]
Let $Tt=xH+\rho$ with $0\leqslant\rho<H$.  If $\varepsilon=1$, all groups at
positions below $x$ lie at or before $t$, and a member of $K_\nu(x)$ does so
exactly when $kE\leqslant\rho$.  If $\varepsilon=-1$, all groups through
position $x$ lie at or before $t$, and a member of $K_\nu(x+1)$ does so exactly
when $kE\geqslant H-\rho$.  Thus the local correction of
Definition~\ref{def:virtual-marker-codes} is explicitly
\begin{equation}\label{eq:virtual-local-rank}
 \alpha_\nu(t)=
 \begin{cases}
  \displaystyle\sum_{k\in K_\nu(x)}\mathbf 1[kE\leqslant\rho],
    &\varepsilon=1,\\[2mm]
  \displaystyle\sum_{k\in K_\nu(x+1)}\mathbf 1[kE\geqslant H-\rho],
    &\varepsilon=-1.
 \end{cases}
\end{equation}
The true slot determined by a uniform cell and a local correction is
\begin{equation}\label{eq:virtual-code-to-slot}
 i_\nu(x,\alpha)=
 \begin{cases}
  \operatorname{pref}_\nu(x-1)+\alpha-1,&\varepsilon=1,\\
  \operatorname{pref}_\nu(x)+\alpha-1,&\varepsilon=-1.
 \end{cases}
\end{equation}
Indeed, the right-hand side is the number of true boundaries at or before
$t$, minus one.  Since $e_0=0$, every code attained by a marker gives a valid
index and~\eqref{eq:virtual-code-to-slot} equals $\iota(t)$.

For a uniform rectangle $(x,y)$ and a valid correction pair define
\begin{equation}\label{eq:virtual-menu-entry}
 \mathcal M_{\nu,P}(x,y)[\alpha_u,\alpha_v]
 =\bigl(i_\nu(x,\alpha_u),i_\nu(y,\alpha_v),
   \mathsf{Op}_{P;i_\nu(x,\alpha_u),i_\nu(y,\alpha_v)}\bigr),
\end{equation}
 together with the affine terminal-marker data of the selected true slot pair.
Equation~\eqref{eq:virtual-code-to-slot} depends only on the path, the chosen
grid, and the two local corrections.  Therefore the menu may be constructed before
any stopped terminal pair is processed.  The group bounds proved above leave
at most four values of each local correction, hence at most sixteen entries per
uniform rectangle.  Equations~\eqref{eq:virtual-local-rank} and
\eqref{eq:virtual-code-to-slot} prove that the selected entry is the exact
true rectangle and operator.  This proves
Theorem~\ref{thm:virtual-rectangle-menu}.

\medskip
\noindent\textbf{Construction cost.}
Finally $H_*=2\pi_{11}+\pi_{12}\leqslant3\tau$.  The $L$ sample residues can therefore be
placed in an array of length $H_*$ and read in order in $O(\tau)$ time and
words.  Insert each ordered orbit index into
$K_0(\lambda_j)$ and $K_1(\lambda_j+\omega_j)$, then scan the two arrays to
form their prefix counts.  Their total length is
$(\pi_{11}+1)+(\pi_{11}+\pi_{12}+1)
=2\pi_{11}+\pi_{12}+2=O(\tau)$, and the group bounds above give constant-size
entries.  One fold trace per slot costs $O(Ld)$; because every nonempty prefix
contributes at least one to $L$, we have $d<L$ and thus $O(Ld)=O(L^2)$.

Finally enumerate all coarse pairs and their at most sixteen local-code pairs
using~\eqref{eq:virtual-code-to-slot}.  The two complete menu sizes are
\[
 16\pi_{11}^2+16(\pi_{11}+\pi_{12})^2=O(\tau^2).
\]
This is also their construction time and proves the final cost assertion of
Theorem~\ref{thm:virtual-rectangle-menu}.

\subsection{Constant-size operator representation}
\label{app:operator-representation}

For one fixed query, pointwise evaluation uses the six lattice moments
together with the four univariate power sums
$\operatorname{Pow}_0,\ldots,\operatorname{Pow}_3$.  For compilation, however,
one operator must be applicable before the root lengths
$q$ and $h$ are known.  The complement of the transformed child region must
therefore be retained as a boundary-correction polynomial in $(q,h)$.
Summing a degree-two lattice monomial over a rectangle can raise the total
boundary degree to four; for example,
\[
\sum_{0\leqslant t<q}\sum_{1\leqslant s\leqslant h}t^2
=h\,\frac{q(q-1)(2q-1)}6.
\]
The degree bounds for the six corrections are
\[
2,\quad3,\quad4,\quad3,\quad4,\quad4.
\]
We store each correction on the corresponding initial part of the monomial
triangle
\[
q^ih^j,\qquad i,j\geqslant0,\quad i+j\leqslant4,
\]
so the coefficient counts are $6,10,15,10,15,15$.  The number $15$ is the
dimension of the full bivariate degree-four correction space, not the size of
the recursive moment kernel, which remains six.  Under one cycle the update is
$(q,h)\mapsto(h,q-Ah-d)$ with $d=U+V+2-A$.  Substitution of these affine forms
into the corrections shows that the full degree-four triangle is closed.

\medskip
\noindent\textbf{Executable one-cycle coefficients.}
The following formulas remove any need to reconstruct the operator from the
geometric argument.  Put
\[
 S_r(h):=\sum_{k=0}^{h-1}k^r,
 \qquad
 S^+_{r,j}(h):=\sum_{k=0}^{h-1}k^r(k+1)^j
 =\sum_{\ell=0}^{j}\binom{j}{\ell}S_{r+\ell}(h).
\]
Only $S_0,S_1,S_2,S_3$ occur, with
\[
 S_0=h,\quad S_1=\frac{h(h-1)}2,\quad
 S_2=\frac{h(h-1)(2h-1)}6,\quad S_3=S_1^2.
\]
For $h\geqslant1$, this notation is related to the earlier power sums by
$S_0(h)=1+\Sigma_0(h-1)$ and
$S_r(h)=\Sigma_r(h-1)$ for $r\geqslant1$; all these sums are zero at $h=0$.
For $j=0,1,2$ define $E_{0j}$ by the first line below; for $j=0,1$ define
$E_{1j}$ by the second line; finally define $E_{20}$ by the remaining lines:
\begin{align}
 E_{0j}(q,h;A,U)
  &=(q-U-1)S^+_{0,j}-A S^+_{1,j},\label{eq:explicit-E0j}\\
 E_{1j}(q,h;A,U)
  &=\frac12\bigl(q(q-1)S^+_{0,j}-A^2S^+_{2,j}
       -A(2U+1)S^+_{1,j}-U(U+1)S^+_{0,j}\bigr),
       \label{eq:explicit-E1j}\\
 E_{20}(q,h;A,U)
  &=\frac16\bigl(q(q-1)(2q-1)S^+_{0,0}-2A^3S^+_{3,0}\nonumber\\
  &\hspace{16mm}{}-3A^2(2U+1)S^+_{2,0}
       -A(6U^2+6U+1)S^+_{1,0}\nonumber\\
  &\hspace{16mm}{}-U(U+1)(2U+1)S^+_{0,0}\bigr).
       \label{eq:explicit-E20}
\end{align}
Here every $S^+_{r,j}$ is evaluated at $h$.  In the fixed coordinate order of
\eqref{eq:lattice-staircase}, set
\[
 \mathbf E(q,h;A,U)=
 (E_{00},E_{10},E_{20},E_{01},E_{11},E_{02})^{\mathsf T}.
\]
Expanding the affine image in~\eqref{eq:affinecoords} gives the explicit
one-step moment identity
\begin{equation}\label{eq:explicit-six-moment-step}
 \mathbf L_{\rm parent}
 =\mathbf E(q,h;A,U)-\mathsf C(A,U)\mathbf L_{\rm child},
\end{equation}
where
\[
\mathsf C(A,U)=
\begin{pmatrix}
1&0&0&0&0&0\\
U&A&0&1&0&0\\
U^2&2AU&A^2&2U&2A&1\\
1&1&0&0&0&0\\
U&A+U&A&1&1&0\\
1&2&1&0&0&0
\end{pmatrix}.
\]
Equations~\eqref{eq:explicit-E0j}--\eqref{eq:explicit-six-moment-step}
are the six scalar formulas used by both pointwise evaluation and operator
compilation.

The affine state part is equally explicit.  For
$\widetilde{\mathbf s}=(a,b,q,h,u,v,1)^{\mathsf T}$,
\begin{equation}\label{eq:explicit-state-matrix}
 \widetilde{\mathbf s}'=\mathsf S(A,U,V)\widetilde{\mathbf s},\qquad
\mathsf S(A,U,V)=
\begin{pmatrix}
0&1&0&0&0&0&0\\
1&-A&0&0&0&0&0\\
0&0&0&1&0&0&0\\
0&0&1&-A&0&0&-(U+V+2-A)\\
0&U+1&0&0&-1&0&-1\\
0&V+1&0&0&0&-1&-1\\
0&0&0&0&0&0&1
\end{pmatrix}.
\end{equation}

Finally, suppose $\mathsf{Op}_1=(\mathsf A_1,\mathsf U_1,\Pi_1)$ maps the
root to an intermediate state and $\mathsf{Op}_2$ maps that state to the
terminal state.  Let $(q_1(q,h),h_1(q,h))$ be the two length coordinates
obtained from $\mathsf A_1$.  Appending the second record means
\begin{equation}\label{eq:explicit-operator-composition}
 \mathsf{Op}_2\circ\mathsf{Op}_1=
 \left(
   \mathsf A_2\mathsf A_1,
   \mathsf U_1\mathsf U_2,
   \Pi_1(q,h)+\mathsf U_1
      \Pi_2(q_1(q,h),h_1(q,h))
 \right).
\end{equation}
Thus composition consists only of two fixed matrix products and one affine
substitution in the fifteen degree-four monomials.  Formula
\eqref{eq:explicit-operator-composition} determines the order of multiplication
in this composition.

\medskip
\noindent\textbf{Optimized pointwise six-moment kernel.}
Short records and the optional terminal reciprocal step use the following
floor-moment recurrence directly.  To include the unnormalized calls produced
inside the recursion, write
$f(t)=\floor{(ct+\gamma)/m}$ on $0\leqslant t<n$.  Let
$\Phi^f_{ij}:=\Phi_{ij}(f)$, and use $\Phi^g_{ij}$ for the corresponding
moments of a child query $g$.  Thus
\[
 (\Phi^f_{01},\Phi^f_{11},\Phi^f_{21},
   \Phi^f_{02},\Phi^f_{12},\Phi^f_{03})
 =\sum_{t=0}^{n-1}(f,t f,t^2f,f^2,t f^2,f^3).
\]
If $n=0$, return six zeros before performing a division.
Reuse the already defined power sums $S_r(n)=\sum_{t=0}^{n-1}t^r$.
If $c=Q_c m+r$, $\gamma=Q_\gamma m+z$, and
$g(t)=\floor{(rt+z)/m}$, then
$f(t)=Q_ct+Q_\gamma+g(t)$ and
\begin{align}
\Phi^f_{01}&=Q_cS_1(n)+Q_\gamma n+\Phi^g_{01},\nonumber\\
\Phi^f_{11}&=Q_cS_2(n)+Q_\gamma S_1(n)+\Phi^g_{11},\nonumber\\
\Phi^f_{21}&=Q_cS_3(n)+Q_\gamma S_2(n)+\Phi^g_{21},\nonumber\\
\Phi^f_{02}&=Q_c^2S_2(n)+2Q_cQ_\gamma S_1(n)+Q_\gamma^2n
     +2Q_c\Phi^g_{11}+2Q_\gamma\Phi^g_{01}+\Phi^g_{02},
\nonumber\\
\Phi^f_{12}&=Q_c^2S_3(n)+2Q_cQ_\gamma S_2(n)+Q_\gamma^2S_1(n)
\nonumber\\
&\quad+2Q_c\Phi^g_{21}+2Q_\gamma\Phi^g_{11}+\Phi^g_{12},
\nonumber\\
\Phi^f_{03}&=Q_c^3S_3(n)+3Q_c^2Q_\gamma S_2(n)
\nonumber\\
&\quad+3Q_cQ_\gamma^2S_1(n)+Q_\gamma^3n
\nonumber\\
&\quad+3Q_c^2\Phi^g_{21}+6Q_cQ_\gamma\Phi^g_{11}
       +3Q_\gamma^2\Phi^g_{01}
\nonumber\\
&\quad+3Q_c\Phi^g_{12}+3Q_\gamma\Phi^g_{02}+\Phi^g_{03}.
\label{eq:floor-kernel-affine}
\end{align}
If $r=0$, then $g$ is identically zero and its six moments vanish.
Otherwise replace $(c,\gamma)$ by $(r,z)$; hence the remaining branch has
$0<c<m$ and $0\leqslant\gamma<m$.  Put
\[
 h=\floor{\frac{c(n-1)+\gamma}{m}},\qquad
 g(k)=\floor{\frac{mk+m+c-1-\gamma}{c}},\quad0\leqslant k<h.
\]
Discrete reciprocity gives
\begin{align}
\Phi^f_{01}&=nh-\Phi^g_{01},\nonumber\\
\Phi^f_{11}&=hS_1(n)-\frac{\Phi^g_{02}-\Phi^g_{01}}{2},\nonumber\\
\Phi^f_{21}&=hS_2(n)
 -\frac{2\Phi^g_{03}-3\Phi^g_{02}+\Phi^g_{01}}{6},\nonumber\\
\Phi^f_{02}&=nh^2-2\Phi^g_{11}-\Phi^g_{01},\nonumber\\
\Phi^f_{12}&=h^2S_1(n)-\Phi^g_{12}-\frac{\Phi^g_{02}}{2}
 +\Phi^g_{11}+\frac{\Phi^g_{01}}{2},
\label{eq:floor-kernel-reciprocal}\\
\Phi^f_{03}&=nh^3-3\Phi^g_{21}-3\Phi^g_{11}-\Phi^g_{01}.
\nonumber
\end{align}
If $h=0$, the reciprocal child has empty range and again contributes six
zeros.  Otherwise it is evaluated recursively with parameters
$(m,m+c-1-\gamma,c,h)$ in the order (coefficient, intercept, modulus,
length).
The first branch reduces $(c,\gamma)$ modulo $m$ and the second exchanges
the two positive coefficients, so the recursion has ordinary Euclidean depth.
All six outputs are formed from shared powers and power sums; no generic
polynomial evaluator is needed.

The universal one-step formulas use power sums whose denominators divide
$\operatorname{lcm}(2,4,6)=12$, followed by the fixed conversion to lattice
moments.  Expanding them shows that every
coefficient has denominator dividing $72$.  We therefore store the integral
coefficient vector of $72\Pi_{P;i,j}$.  Integral affine substitution preserves this
scale, and application performs one exact division by $72$ per returned
component.

\begin{proof}[Full proof of Lemma~\ref{lem:closure}]
Equation~\eqref{eq:cycle} is an integral affine state map, so homogeneous
affine matrices are closed under composition.  The basis
$1,t,t^2,s,ts,s^2$ is closed under an affine substitution in two variables,
so~\eqref{eq:affinecoords} induces a fixed six-dimensional map.
In~\eqref{eq:one-step-lattice}, sums of degree-two monomials over the explicit
regions have total boundary degree at most four.  Since the child staircase dimensions
are affine in $(q,h)$ by~\eqref{eq:cycle}, induction preserves this degree.
The one-step power sums have common denominator dividing
$\operatorname{lcm}(2,4,6)=12$, and the triangular moment conversion adds a
divisor of $6$; hence $72$ is a common scale.  Integral affine substitution
preserves it.  The constant-size state matrix, six affine coordinate
coefficients, one sign, and fifteen scaled monomial coefficients per component
are therefore sufficient.
\end{proof}

\subsection{Cell-operator and periodic-table bounds}
\label{app:operator-proofs}

\begin{proof}[Proof of Lemma~\ref{lem:tree}]
Let $n_j$ be the number of one-marker groups after the first $j$ coefficient
steps.  Pairing the groups leaves at most $n_j^2$ active operator states at
that step.  Apply the boundary construction to the prefix matrix $P_j$.  If
$\chi_j=(P_j)_{11}$ and $\chi_{j-1}=(P_j)_{12}$, then
\eqref{eq:orbit-count-bound} gives
\[
 n_j=1+\sum_{m=1}^{j}\chi_m
 \leqslant2\chi_j+\chi_{j-1}\leqslant3\chi_j.
\]
The continuants satisfy
\[
 \chi_{j+1}=A_j\chi_j+\chi_{j-1}\geqslant\chi_j+\chi_{j-1},
\]
so $\chi_{j+2}\geqslant2\chi_j$.  When read backward, each of the even and odd
subsequences $(\chi_j^2)_j$ is dominated by a geometric series; hence
$\sum_j\chi_j^2=O(\chi_d^2)$.  Moreover
$\chi_d\leqslant L$, and the path depth is $O(\log(2+L))$.  Hence
$\sum_jn_j^2=O(L^2)$.  Every active state performs one constant-size operator
composition, and the final $L^2$ states produce the cell operators.
\end{proof}

\subsubsection{Small-modulus periodic table}
\label{app:periodic-table}

\begin{proof}[Proof of Lemma~\ref{lem:periodic}]
If $q_{\rm term}=0$, $h_{\rm term}=0$, or $R=0$, the terminal staircase is
empty and all six moments vanish.  Assume below that
$q_{\rm term},h_{\rm term},R>0$.

\noindent\emph{Gcd lookup.}
Store the triangular table
\[
 \mathsf{GCD}[s,0]=s,\qquad
 \mathsf{GCD}[s,r]=\mathsf{GCD}[r,s\bmod r]
 \quad(1\leqslant r<s\leqslant\kappa).
\]
Constructing rows in increasing $s$ makes every entry on the right available
when it is needed, and the Euclidean recurrence gives
$\mathsf{GCD}[s,r]=\gcd(s,r)$.  Hence the table uses $O(\kappa^2)$ time and entries.
Every query passed directly to the table has $M\leqslant\kappa$ and uses
$\mathsf{GCD}[M,R]$.

\smallskip
\noindent\emph{Optional reciprocal cycle.}
If $M>\kappa$ for the terminal query,
Lemma~\ref{lem:onestep} gives $0<R<\kappa$.  Apply one reciprocal cycle as in
\eqref{eq:child-query}; its child coefficient pair is
$(R,M\bmod R)$, with both entries below $\kappa$.  A zero second entry gives
zero child moments; otherwise the child is evaluated by the table procedure
below, including reduction by $\mathsf{GCD}[R,M\bmod R]$.  In both cases,
after conversion to lattice moments, identity~\eqref{eq:one-step-lattice}
lifts the child moments to those
of the original terminal query in $O(1)$ operations.  Thus it remains to
describe the lookup for a query whose coefficients are at most $\kappa$.

\smallskip
\noindent\emph{Removing a common divisor.}
Consider a nonzero staircase passed to the table,
\[
 f(t)=\floor{\frac{Rt+\beta}{M}},
 \qquad g=\gcd(M,R),
 \qquad M=g\bar M,\quad R=g\bar R,
\]
and write $\beta=gc_\beta+r_\beta$, where $0\leqslant r_\beta<g$.  Then
\begin{equation}\label{eq:periodic-gcd-reduction}
 f(t)=\floor{\frac{\bar R t+c_\beta}{\bar M}}.
\end{equation}
Indeed, before taking the floor the omitted term is
$r_\beta/(g\bar M)<1/\bar M$, while every fractional part of
$(\bar R t+c_\beta)/\bar M$ is an integral multiple of $1/\bar M$.  The omitted
term therefore cannot cross the next integer.  This proves
\eqref{eq:periodic-gcd-reduction} and reduces the coefficients to the coprime
pair $(\bar M,\bar R)$.

\smallskip
\noindent\emph{Removing the intercept.}
For this coprime pair define the zero-intercept staircase
$f_0(t)=\floor{\bar R t/\bar M}$.  Choose the unique
$\zeta\in\{0,\ldots,\bar M-1\}$ satisfying
\[
 \zeta\equiv c_\beta\bar R^{-1}\pmod{\bar M},
 \qquad \bar R\zeta=z_0\bar M+c_\beta.
\]
Equation~\eqref{eq:periodic-gcd-reduction} becomes
\begin{equation}\label{eq:periodic-phase-shift}
 f(t)=f_0(t+\zeta)-z_0.
\end{equation}
Thus every intercept for the pair $(\bar M,\bar R)$ is obtained by a cyclic shift
and a vertical translation of the same zero-intercept staircase.

\smallskip
\noindent\emph{Stored prefixes and one query.}
For each coprime $(\bar M,\bar R)$ store
$\bar R^{-1}\bmod\bar M$ and the prefix vectors
\[
 \operatorname{Pref}(s)=\sum_{0\leqslant t<s}
 \bigl(f_0(t),t f_0(t),t^2f_0(t),f_0(t)^2,
       t f_0(t)^2,f_0(t)^3\bigr),
 \qquad 0\leqslant s\leqslant2\bar M.
\]
These are precisely the six floor-power moments in~\eqref{eq:sixmom}.
The second period is needed because a shifted fragment may start at
$\zeta$ and cross the end of the first period.  For a requested length
$q_{\rm req}$, write $q_{\rm req}=w\bar M+\ell$, where
$w=\floor{q_{\rm req}/\bar M}$ and $0\leqslant\ell<\bar M$.  The final fragment
is the difference of two stored prefix vectors.  The $w$ complete periods use
\[
 f_0(t+j\bar M)=f_0(t)+j\bar R.
\]
Here are the six constant-time shift formulas explicitly.  Extend the stored
prefix notation by
\[
 \operatorname{Pref}_{pr}(s):=\sum_{0\leqslant t<s}t^p f_0(t)^r,
 \qquad (p,r)\in
 \{(0,1),(1,1),(2,1),(0,2),(1,2),(0,3)\},
\]
where the six entries with $r>0$ are stored and
$\operatorname{Pref}_{p0}(s):=S_p(s)$ is an ordinary power sum.  For
$0\leqslant\ell\leqslant\bar M$ define the phase-corrected fragment
moments
\begin{equation}\label{eq:periodic-fragment-binomial}
 \operatorname{Frag}_{pr}^{(\ell)}=
 \sum_{i=0}^{p}\sum_{j=0}^{r}
 \binom pi\binom rj(-\zeta)^{p-i}(-z_0)^{r-j}
 \bigl(\operatorname{Pref}_{ij}(\zeta+\ell)
       -\operatorname{Pref}_{ij}(\zeta)\bigr).
\end{equation}
Thus $\operatorname{Frag}_{pr}^{(\ell)}=\sum_{0\leqslant s<\ell}
s^p(f_0(s+\zeta)-z_0)^r$.  Put
$S_1(\ell)=\ell(\ell-1)/2$ and
$S_2(\ell)=\ell(\ell-1)(2\ell-1)/6$, consistently with the power sums
already used above.  The block beginning at horizontal offset $j\bar M$ has
\begin{align}
 \operatorname{Block}_{01}(j;\ell)
 &=\operatorname{Frag}_{01}^{(\ell)}+j\bar R\ell,
   \label{eq:periodic-B01}\\
 \operatorname{Block}_{11}(j;\ell)
 &=\operatorname{Frag}_{11}^{(\ell)}
   +j\bar M\operatorname{Frag}_{01}^{(\ell)}
   +j\bar R S_1(\ell)+j^2\bar M\bar R\ell,
   \label{eq:periodic-B11}\\
 \operatorname{Block}_{21}(j;\ell)
 &=\operatorname{Frag}_{21}^{(\ell)}
   +2j\bar M\operatorname{Frag}_{11}^{(\ell)}
   +j^2\bar M^2\operatorname{Frag}_{01}^{(\ell)}
   +j\bar R S_2(\ell)\nonumber\\
 &\quad{}+2j^2\bar M\bar R S_1(\ell)
   +j^3\bar M^2\bar R\ell,\label{eq:periodic-B21}\\
 \operatorname{Block}_{02}(j;\ell)
 &=\operatorname{Frag}_{02}^{(\ell)}
   +2j\bar R\operatorname{Frag}_{01}^{(\ell)}
   +j^2\bar R^2\ell,\label{eq:periodic-B02}\\
 \operatorname{Block}_{12}(j;\ell)
 &=\operatorname{Frag}_{12}^{(\ell)}
   +j\bar M\operatorname{Frag}_{02}^{(\ell)}
   +2j\bar R\operatorname{Frag}_{11}^{(\ell)}
   +2j^2\bar M\bar R\operatorname{Frag}_{01}^{(\ell)}\nonumber\\
 &\quad{}+j^2\bar R^2S_1(\ell)
   +j^3\bar M\bar R^2\ell,\label{eq:periodic-B12}\\
 \operatorname{Block}_{03}(j;\ell)
 &=\operatorname{Frag}_{03}^{(\ell)}
   +3j\bar R\operatorname{Frag}_{02}^{(\ell)}
   +3j^2\bar R^2\operatorname{Frag}_{01}^{(\ell)}
   +j^3\bar R^3\ell.
   \label{eq:periodic-B03}
\end{align}
Consequently the answer to a table query of length
$q_{\rm req}=w\bar M+\ell$ is, componentwise,
\begin{equation}\label{eq:periodic-complete-query}
 \Phi_{pr}=\sum_{j=0}^{w-1}\operatorname{Block}_{pr}(j;\bar M)
             +\operatorname{Block}_{pr}(w;\ell).
\end{equation}
The first sum uses only the four ordinary values
$\sum1,\sum j,\sum j^2,\sum j^3$; the second term uses two stored prefix
vectors through~\eqref{eq:periodic-fragment-binomial}.  This yields an
$O(1)$ lookup algorithm for all six moments.

Both the horizontal shift $t\mapsto t+j\bar M$ and the vertical shift
$f_0\mapsto f_0+j\bar R$ are affine.  Consequently every one of the six block
moments is a polynomial of degree at most three in $j$, and the sum over all
complete periods is obtained from the ordinary power sums of
$1,j,j^2,j^3$.  Equation~\eqref{eq:periodic-phase-shift} is one further
horizontal and vertical affine shift; any pure powers introduced by the
translation are ordinary power sums.  Thus the table query is answered with a
constant number of prefix lookups and arithmetic operations, and
\eqref{eq:floor-lattice-conversion} converts the result to the six lattice
moments.  If the optional reciprocal cycle was taken, the constant-size lift
described above returns the moments of the original terminal query expected by
the rectangle operator.

\smallskip
\noindent\emph{Construction size.}
For a fixed $\bar M$ there are $\varphi(\bar M)$ possible coprime values of
$\bar R$, and each stores $2\bar M+1$ constant-size vectors.  Hence the total
number of stored vectors is
\[
\sum_{2\leqslant \bar M\leqslant\kappa}2\bar M\varphi(\bar M)=O(\kappa^3).
\]
All prefixes take the same $O(\kappa^3)$ construction time.  Computing every
stored modular inverse separately by extended Euclid costs
$O(\kappa^2\log\kappa)=O(\kappa^3)$ more operations.  Together with the
$O(\kappa^2)$ gcd table, this proves all three claims of
Lemma~\ref{lem:periodic}.
\end{proof}

\subsection{Short-query fallback}
\label{app:early-queries}

\begin{proof}[Proof of Theorem~\ref{thm:early}]
At reciprocal state $i$, let $\beta_i=a_i-u_i-1$ be its intercept and write
the length pair as $\ell_i=(q_i,h_i)^T$, the coefficient quotient as $A_i$, the two marker
quotients as $U_i,V_i$, and $Q_i:=Q(A_i)$.  Let
$e_1=(1,0)^T$ and $e_2=(0,1)^T$.  Equation~\eqref{eq:cycle} reversed is
\begin{equation}\label{eq:lengthreverse}
 \ell_i=Q_i\ell_{i+1}+d_i e_1,
 \qquad d_i=U_i+V_i+2-A_i\leqslant A_i+2.
\end{equation}
The bound follows from $0\leqslant u_i,v_i<a_i$, which gives
$0\leqslant U_i,V_i\leqslant A_i$.
Let $P_i=Q_0\cdots Q_{i-1}$, let $\mathbf c_i=P_i e_1$ be its first column,
and put $\mathbf c_{-1}=e_2$.  Then, componentwise,
\begin{equation}\label{eq:continuantrec}
 \mathbf c_{i+1}=A_i\mathbf c_i+\mathbf c_{i-1},
 \qquad
 \sum_{i=0}^{k-1}\mathbf c_i\leqslant
 \mathbf c_k+\mathbf c_{k-1}.
\end{equation}

Suppose that a query would reach zero height before the end of its compiled
coefficient path, and let $k$ be the first state with $h_k=0$.  Then
$\floor{(b_k(q_k-1)+\beta_k)/a_k}=0$, hence $q_k\leqslant A_k+1$.  Because the
coefficient path continues, the accepted columns
$\mathbf c_{k-1},\mathbf c_k,\mathbf c_{k+1}$ are componentwise at most
$\tau$.  Iterating
\eqref{eq:lengthreverse} and using~\eqref{eq:continuantrec} gives
\[
 \ell_0=P_k\ell_k+\sum_{i<k}d_i\mathbf c_i,\qquad
 P_k\ell_k=q_k\mathbf c_k
 \leqslant(A_k+1)\mathbf c_k
 \leqslant \mathbf c_{k+1}+\mathbf c_k
 \leqslant2\tau,
\]
and
\[
 \sum_{i<k}d_i\mathbf c_i
 \leqslant\sum_{i<k}(A_i+2)\mathbf c_i
 \leqslant3(\mathbf c_k+\mathbf c_{k-1})
 \leqslant3\mathbf c_{k+1}\leqslant3\tau.
\]
For the middle inequality, the recurrence telescopes to
$\sum_{i<k}A_i\mathbf c_i
=\mathbf c_k+\mathbf c_{k-1}-\mathbf c_0-\mathbf c_{-1}$, while
$2\sum_{i<k}\mathbf c_i\leqslant
2(\mathbf c_k+\mathbf c_{k-1})$ by
\eqref{eq:continuantrec}.
Therefore both root lengths satisfy $q_0,h_0\leqslant5\tau$, so every query
that would terminate along its compiled path is short.  Pointwise evaluation
has ordinary Euclidean depth $O(\log a)=O(\log D)=O(\log(2+\tau))$.

Recall $q_{\mathrm{sq}}=\floor{(N-a^2)/b}$ from
Appendix~\ref{app:half-domain-moments}.  If
$q_{\mathrm{sq}}\geqslant a-1$, the bound $q\leqslant5\tau$ leaves only
$O(\tau^2)$ pairs $(a,b)$.  Otherwise $q=q_{\mathrm{sq}}\leqslant5\tau$ and
\[
 N-a^2<b(5\tau+1)\leqslant D(5\tau+1).
\]
Since $D^2\leqslant N$, this implies $D-a<5\tau+1$.  Hence there are
$O(\tau)$ possible moduli and fewer than $D$ coefficients for each.
Altogether there are $O(D\tau)$ short queries; multiplying by the
pointwise-recursion depth proves the theorem.
\end{proof}

\subsection{Complete \texorpdfstring{$O(n\log n)$}{O(n log n)} algorithm:
formulas and pseudocode}
\label{app:complete-algorithm}

Fix a layer size $N$.  For a direction pair $(a,b)$, the half-domain
reduction produces the normalized query~\eqref{eq:floor-query}.  Its length
$q$ is the number of staircase columns.  Write
\[
 N-bq=\eta a+\beta,
 \qquad 0\leqslant\beta<a.
\]
The quotient $\eta$ is needed only when the staircase moments are converted
back to the original side coordinates.

The six moments $J_{ij}$ of~\eqref{eq:halfmoment} split as
$J_{ij}=J^{\mathrm{upper}}_{ij}-J^{\mathrm{lower}}_{ij}$ by
Appendix~\ref{app:half-domain-moments}.  The lower part and the
equality-boundary correction $E_{a,b}$ are explicit power-sum expressions,
whereas the
upper part contains the recursive floor query.  Hence record generation adds
the known contribution
\begin{equation}\label{eq:record-closed-correction}
 \mathbf c^{\mathrm{closed}}_{a,b}
 :=-8\mathcal A_{a,b}(J^{\mathrm{lower}})-E_{a,b},
\end{equation}
and, after the selected cell operator has returned the six staircase moments,
adds
\begin{equation}\label{eq:record-moment-contribution}
 \mathbf c^{\mathrm{query}}_{a,b}
 :=8\mathcal A_{a,b}(J^{\mathrm{upper}}).
\end{equation}
Their sum is exactly the contribution of $(a,b)$ to the layer triple.  A
nonempty record stores the root state of
Definition~\ref{def:affine-query-state}, its stopped terminal pair $(M,R)$,
and $\eta$.  The pointwise recursion recovers its intercept from
$\beta=a-u-1$.

For every integer $x\geqslant0$, define the three weighted M\"obius prefixes
\[
 \mathfrak M_s(x)=\sum_{d\leqslant x}\mu(d)d^s,
 \qquad s=0,1,2.
\]
Thus $\mathfrak M_s(0)=0$.  For $x\geqslant1$, the identity
$\sum_{d\mid k}\mu(d)=[k=1]$ gives
\begin{equation}\label{eq:mobius-prefix}
 \mathfrak M_s(x)=1-\sum_{j=2}^{x}j^s\mathfrak M_s(\floor{x/j}).
\end{equation}
Equal quotients in this recurrence are grouped into constant-time blocks by
the ordinary power sums.  If $\kappa(D)$ denotes the terminal threshold for
layer cap $D$, put
\[
 \kappa_{\max}(n)=
 \max_{1\leqslant D\leqslant\floor{\sqrt n}}\kappa(D).
\]
It is obtained by one scan over these $D$ values; this costs
$O(\sqrt n)$ arithmetic operations and $O(1)$ additional words.
The periodic table is constructed once at limit $\kappa_{\max}(n)$ and shared
by all layer calls.

\begin{algorithm}[H]
  \small
  \caption{Quotient-blocked $O(n\log n)$ algorithm.}
  \label{alg:rectangle-count}
  \begin{algorithmic}
    \AlgLine{0}{\textbf{procedure} $\textsc{RectangleCount}(n)$}
    \AlgLine{1}{initialize $\mathfrak M_0,\mathfrak M_1,\mathfrak M_2$
      by~\eqref{eq:mobius-prefix}}
    \AlgLine{1}{construct the read-only periodic table
      $\mathsf T_{\kappa_{\max}(n)}$}
    \AlgLine{1}{$ans\gets F_0(n)$ from~\eqref{eq:boundary-count}; $\ell\gets1$}
    \AlgLine{1}{\textbf{while} $\ell\leqslant n$ \textbf{do}}
    \AlgLine{2}{$N\gets\floor{n/\ell}$;
      $r_{\rm end}\gets\floor{n/N}$}
    \AlgLine{2}{\textbf{for} $s=0,1,2$ \textbf{do}
      $\Delta_s\gets\mathfrak M_s(r_{\rm end})-\mathfrak M_s(\ell-1)$}
    \AlgLine{2}{\textbf{if} $(\Delta_0,\Delta_1,\Delta_2)\ne(0,0,0)$
      \textbf{then}}
    \AlgLine{3}{$(A_N,B_N,C_N)\gets
      \textsc{RecursiveLayer}(N,\mathsf T_{\kappa_{\max}(n)})$}
    \AlgLine{3}{$ans\gets ans+n^2A_N\Delta_0
      -nB_N\Delta_1+C_N\Delta_2$}
    \AlgLine{2}{\textbf{end if}; $\ell\gets r_{\rm end}+1$}
    \AlgLine{1}{\textbf{end while}; \textbf{return} $ans$}
    \AlgLine{0}{\textbf{end procedure}}
  \end{algorithmic}
\end{algorithm}

\begin{algorithm}[H]
  \small
  \caption{Batched divisor layer.}
  \label{alg:linear-layer}
  \begin{algorithmic}
    \AlgLine{0}{\textbf{procedure} $\textsc{RecursiveLayer}(N,\mathsf T)$}
    \AlgLine{1}{compute $(D,\sigma,\tau,\kappa)$ by~\eqref{eq:layer-scales}}
    \AlgLine{1}{$Z\gets\mathcal D(N)$ from~\eqref{eq:diagonal-contribution}}
    \AlgLine{1}{traverse the coefficient-cone recursion $P\mapsto PQ(A)$}
    \AlgLine{1}{\textbf{for every reached cone} $P$ \textbf{do}}
    \AlgLine{2}{construct the symbolic true-marker slots and their traces as in
      Appendix~\ref{app:marker-orbit-proof}}
    \AlgLine{2}{compile all $\mathsf{Op}_{P;i,j}$ through the shared prefix tree}
    \AlgLine{2}{construct both uniform grids and their correction data}
    \AlgLine{2}{\textbf{for every stopped pair} $(M,R)\in\mathcal T(P)$
      \textbf{do}}
    \AlgLine{3}{$\binom ab\gets P\binom MR$; compute
      $\mathbf c^{\mathrm{closed}}_{a,b}$; $Z\gets Z+\mathbf c^{\mathrm{closed}}_{a,b}$;
      generate the possible query record}
    \AlgLine{3}{\textbf{if} $q>0$ \textbf{then}}
    \AlgLine{4}{\textbf{if} $q\leqslant5\tau$ \textbf{then}}
    \AlgLine{5}{obtain $\boldsymbol\Phi$ directly from
      \eqref{eq:floor-kernel-affine}--\eqref{eq:floor-kernel-reciprocal}}
    \AlgLine{4}{\textbf{else} compute the two marker codes, derive $(i,j)$ by
      \eqref{eq:virtual-code-to-slot}, and access $\mathsf{Op}_{P;i,j}$}
    \AlgLine{5}{evaluate the affine terminal coordinates; obtain
      $\mathbf L_{\mathrm{term}}$ from $\mathsf T$}
    \AlgLine{5}{$\mathbf L_{\rm root}\gets
       \mathsf U_{P;i,j}\mathbf L_{\mathrm{term}}+
       \Pi_{P;i,j}(q,h)$; convert $\mathbf L_{\rm root}$ to
       $\boldsymbol\Phi$ by~\eqref{eq:floor-lattice-conversion}}
    \AlgLine{4}{\textbf{end if}; form $J^{\mathrm{upper}}$ from
      $\boldsymbol\Phi$ and $\eta$ by~\eqref{eq:upper-moment-map}}
    \AlgLine{4}{$Z\gets Z+8\mathcal A_{a,b}(J^{\mathrm{upper}})$}
    \AlgLine{3}{\textbf{end if}}
    \AlgLine{2}{\textbf{end for}; discard the path-local objects}
    \AlgLine{1}{\textbf{end for}; \textbf{return} $Z=(A(N),B(N),C(N))$}
    \AlgLine{0}{\textbf{end procedure}}
  \end{algorithmic}
\end{algorithm}

For correctness, Lemma~\ref{lem:momentreduction} produces exactly the closed
correction and floor-moment record for every direction pair, and
Lemma~\ref{lem:terminal-bijection} assigns that pair to exactly one accepted
coefficient path.  Theorem~\ref{thm:virtual-rectangle-menu} selects its exact
true marker rectangle.  Lemmas~\ref{lem:closure} and~\ref{lem:periodic}
evaluate every long record, while Section~\ref{sec:early-termination}
handles every short record.  Thus Algorithm~\ref{alg:linear-layer} contributes
every term of~\eqref{eq:complete-layer-generation} once.  Finally,
\eqref{eq:mobius-prefix} and quotient blocking give~\eqref{eq:layers}, so
Algorithm~\ref{alg:rectangle-count} returns~\eqref{eq:full-geometric-count}.

\medskip
\noindent\textbf{Executable cone traversal.}
Algorithm~\ref{alg:visit-cone} expands the ``traverse'' line of
Algorithm~\ref{alg:linear-layer}.  The predicate
$\operatorname{Root}(P,M,R)$ abbreviates
\[
 1\leqslant \pi_{21}M+\pi_{22}R
 <\pi_{11}M+\pi_{12}R\leqslant D.
\]
An intersection with this predicate is an integer interval: its endpoints are
obtained by signed ceiling and floor divisions as follows.  For fixed $R$,
$a\leqslant D$ gives
$M\leqslant\floor{(D-\pi_{12}R)/\pi_{11}}$; if $\pi_{21}>0$, then $b\geqslant1$
gives $M\geqslant\ceil{(1-\pi_{22}R)/\pi_{21}}$; and if
$\pi_{11}>\pi_{21}$, then $b<a$ gives
\[
 M\geqslant
 \floor{\frac{(\pi_{22}-\pi_{12})R}{\pi_{11}-\pi_{21}}}+1.
\]
If $\pi_{21}=0$, the condition $b\geqslant1$ is checked directly; if
$\pi_{11}=\pi_{21}$, the condition $b<a$ reduces to
$\pi_{22}R<\pi_{12}R$.  Empty intervals are ignored.  Procedure
$\textsc{Emit}$ executes the complete inner-loop action of
Algorithm~\ref{alg:linear-layer} for one stopped pair: it adds the closed
correction and, when the query is nonempty, evaluates either its short or its
long record.  The traversal starts with $\textsc{VisitCone}(I,\bot)$, where
$I$ is the identity matrix and the absent last quotient is never read.

\begin{algorithm}[H]
  \small
  \caption{Depth-first visit of one coefficient cone.}
  \label{alg:visit-cone}
  \begin{algorithmic}
    \AlgLine{0}{\textbf{procedure} $\textsc{VisitCone}(P,A_{\rm last})$}
    \AlgLine{1}{$e\gets\operatorname{ext}(P)$}
    \AlgLine{1}{construct the path slots, traces, operators, grids, and correction data}
    \AlgLine{1}{\textbf{for} $1\leqslant M\leqslant\tau$ and
      $0\leqslant R<M$ \textbf{do}}
    \AlgLine{2}{\textbf{if} $\operatorname{Root}(P,M,R)$ and
      ($P=I$ or $A_{\rm last}M+R>\tau$) \textbf{then}
      $\textsc{Emit}(P,M,R)$}
    \AlgLine{1}{\textbf{end for}}
    \AlgLine{1}{\textbf{if} $P\ne I$ and $A_{\rm last}\geqslant2$
      \textbf{then}}
    \AlgLine{2}{\textbf{for} $\tau<M\leqslant\floor{D/\pi_{11}}$ with
      $\operatorname{Root}(P,M,0)$ \textbf{do} $\textsc{Emit}(P,M,0)$}
    \AlgLine{1}{\textbf{end if}; $A_{\rm stop}\gets e+1$}
    \AlgLine{1}{\textbf{for} $1\leqslant R\leqslant
      \floor{D/(\pi_{11}A_{\rm stop}+\pi_{12})}$ \textbf{do}}
    \AlgLine{2}{$\mathcal I_R\gets
      [\max\{\tau+1,A_{\rm stop}R,R+1\},
      \floor{(D-\pi_{12}R)/\pi_{11}}]$}
    \AlgLine{2}{$\mathcal I_R\gets\mathcal I_R
      \cap\{M:\operatorname{Root}(P,M,R)\}$}
    \AlgLine{2}{\textbf{for} $M\in\mathcal I_R\cap\mathbb Z$ \textbf{do}
      $\textsc{Emit}(P,M,R)$}
    \AlgLine{1}{\textbf{end for}}
    \AlgLine{1}{discard the path-local objects}
    \AlgLine{1}{\textbf{for} $1\leqslant A\leqslant e$ \textbf{do}
      $\textsc{VisitCone}(PQ(A),A)$}
    \AlgLine{0}{\textbf{end procedure}}
  \end{algorithmic}
\end{algorithm}

The three emission blocks are respectively the small-modulus,
zero-remainder, and overflow pieces of $\mathcal T(P)$.  The recursive calls
are precisely the accepted continuations, so
Lemma~\ref{lem:terminal-bijection} proves that the visitor emits every root
pair once.  At a visited path, the path-local setup is also explicit: order
the $L$ sample residues used in~\eqref{eq:orbit-boundaries}; fold one representative
of every nonempty slot to obtain its marker word; pair the two word lists in a
prefix tree and compose the corresponding cycle operators; then form the
groups $K_0,K_1$ and their prefix counts.  The implementation evaluates
\eqref{eq:virtual-local-rank}--\eqref{eq:virtual-code-to-slot} on demand instead
of materializing~\eqref{eq:virtual-menu-entry}.  Only after this setup does
$\textsc{Emit}$ stream the stopped pairs of that path.

\medskip
\noindent\textbf{Worked record.}
Take the layer $N=100$.  Then
\[
 (D,\sigma,\tau,\kappa)=(10,4,2,5).
\]
For $(a,b)=(4,3)$, the half-domain formulas give
\[
 q_{\rm sq}=28,\qquad q=13,\qquad
 N-bq=61=15a+1,
\]
so $\eta=15$, $\beta=1$, and the root state is
\[
 (a,b;q,h;u,v)=(4,3;13,9;2,1).
\]
Its first coefficient quotient is one.  Thus
$P=Q(1)=\left(\begin{smallmatrix}1&1\\1&0\end{smallmatrix}\right)$ is
accepted, the terminal pair is $(M,R)=(3,1)$, and the next quotient three is
rejected because $PQ(3)$ has an entry four, exceeding $\tau$.

Here $L=1+\chi_1=2$ and the true boundary residues are
$0$ and $3$ modulo $H=4$.  Since $2R\leqslant M$, grid zero is selected;
it has $T=\pi_{11}=1$ cell.  Both markers have code $(0,0)$, so the code-to-slot
map selects the unique nonempty rectangle containing $(u,v)=(2,1)$ and its
operator with quotient triple $(A,U,V)=(1,0,0)$.  One cycle gives
\[
 (3,1;9,3;0,1).
\]
Because $M=3\leqslant\kappa$, the terminal table answers
$f'(k)=\floor{(k+2)/3}$ for $0\leqslant k<9$ with
\[
 \mathbf L_{\rm term}=(15,81,507,24,141,46)^{\mathsf T}.
\]
The one-cycle operator~\eqref{eq:one-step-lattice} lifts this to
\[
 \mathbf L_{\rm root}=(57,477,4475,204,1887,996)^{\mathsf T},
\]
and~\eqref{eq:floor-lattice-conversion} gives
\[
 \boldsymbol\Phi=(57,477,4475,351,3297,2433)^{\mathsf T}.
\]
With $\eta=15$, Equation~\eqref{eq:upper-moment-map}, in the order
$J_{00},J_{01},J_{02},J_{10},J_{11},J_{20}$, yields
\[
 J^{\rm upper}=(252,1629,13991,2619,15645,36061).
\]
Hence the deferred contribution is
\[
 8\mathcal A_{4,3}(J^{\rm upper})
 =(2016,237888,7933992).
\]
For completeness, the explicit lower moments are
$(94,823,8287,463,4561,3207)$ in the same order, and
$E_{4,3}=(12,504,5376)$.  Thus
$\mathbf c^{\rm closed}_{4,3}=(-764,-72520,-2021000)$, and the complete
contribution of this record to the layer accumulator is
\[
 (1252,165368,5912992).
\]

\subsection{Detailed \texorpdfstring{$O(n\log n)$}{O(n log n)} complexity accounting}
\label{app:one-value-complexity}

From~\eqref{eq:layer-scales},
\[
 D^2=O(N),\qquad \tau^4=O(N),\qquad
 \kappa^3=O(N^{3/4}).
\]
The scales are linked.  The half-domain reduction needs every pair up to
$D=\floor{\sqrt N}$.  The coefficient-cone and marker-table work is linear
only if $\tau^4=O(D^2)$, hence $\tau=O(\sqrt D)$.  Lemma~\ref{lem:onestep}
needs $\tau\kappa\geqslant D$, so then $\kappa=\Omega(\sqrt D)$.  Thus
$\tau,\kappa=\Theta(\sqrt D)$ is the common boundary that preserves linear
time and minimizes the terminal-table memory within this construction.

\begin{center}
\begin{tabular}{lr}
\toprule
object & arithmetic operations\\
\midrule
coefficient-cone nodes and stopped pairs & $O(D^2+\tau^4)$\\
all boundary grids and correction data & $O(\tau^4)$\\
all cell-operator compilations & $O(\tau^4)$\\
record generation and long-record applications & $O(D^2)$\\
reduced periodic table & $O(\kappa^3)$\\
short-query pointwise evaluations & $O(D\tau\log(2+\tau))$\\
\bottomrule
\end{tabular}
\end{center}

There are $O(\tau^2)$ paths.  Each has $L\leqslant3\tau$, so its true
intervals, $L^2$ operators, and two grids with their prefix data together cost
$O(\tau^2)$; over all paths this is $O(\tau^4)$.  Long records take
constant time and Theorem~\ref{thm:early} gives
$D\tau\log(2+\tau)=o(N)$.  Processing one path at a time uses
$O(\tau^2)$ live words, while the periodic prefixes use $O(\kappa^3)$ words
and dominate.  This completes the detailed proof of Theorem~\ref{thm:layer}.

\subsubsection{Weighted M\"obius prefix computation}
\label{app:mobius-prefix-memory}

Algorithms for isolated values of the unweighted Mertens function provide the
standard context for this quotient-block approach~\cite{deleglise-rivat}; here
the same decomposition is applied to the three weighted prefixes required by
\eqref{eq:mobius-prefix}.  Let $S_\mu$ be a sieve cutoff.  Sieving all arguments up to $S_\mu$ costs
$O(S_\mu)$ time and storage.  The larger distinct arguments have the form
$\floor{n/j}$ with $j\leqslant n/S_\mu$.  Evaluate each such argument once,
in increasing order, and memoize its three weighted prefix values.  If
$x=\floor{n/j}$, every nontrivial recursive child, with $k\geqslant2$, satisfies
\[
 \floor{\floor{n/j}/k}=\floor{n/(jk)}.
\]
It is therefore either covered by the sieve or is another member of the same
memoized distinct-quotient set, and it is already available in this order.
Since the quotient-block recurrence for one argument $x$ has $O(\sqrt x)$
blocks, the total cost after memoization is bounded by
\[
 \sum_{j\leqslant n/S_\mu}\sqrt{n/j}
 =O(nS_\mu^{-1/2}).
\]
Thus the combined bound is $O(S_\mu+n/\sqrt{S_\mu})$.  Balancing its two terms
gives $S_\mu=n/\sqrt{S_\mu}$ and hence $S_\mu=n^{2/3}$.  With
$S_\mu=\ceil{n^{2/3}}$, the prefix computation takes $O(n^{2/3})$ time and
storage.

\begin{proof}[Proof of Lemma~\ref{lem:one-value-operands}]
Root coefficient, marker, and length coordinates are at most $n$.  An
accepted continuant matrix has entries at most $\tau$.  If
$A_0,\ldots,A_{d-1}$ are the coefficient quotients of its path, then
\[
 \prod_i A_i\leqslant\pi_{11}\leqslant\tau,
 \qquad d=O(\log(2+\tau)).
\]
Every local affine-state coefficient is $O(A_i+2)$ because
$0\leqslant U_i,V_i\leqslant A_i$.  Composition in a fixed degree-four
polynomial basis enlarges coefficient magnitudes by at most
$C(A_i+2)^C$ for an absolute constant $C$.  Since every $A_i\geqslant1$,
\[
 \prod_i C(A_i+2)^C
 \leqslant C^d3^{Cd}\Bigl(\prod_iA_i\Bigr)^C
 \leqslant C^d3^{Cd}\tau^C=\tau^{O(1)},
\]
where the last equality uses $d=O(\log(2+\tau))$.  Thus every coefficient
stored in an affine endpoint, terminal marker, moment map, or boundary
correction is $n^{O(1)}$.  More explicitly,
$|\mathfrak M_s(x)|\leqslant\sum_{d\leqslant x}d^s\leqslant n^{s+1}$ for
$s=0,1,2$.  Every floor or periodic moment sums at most $n$ terms whose
coordinates are at most $n$, and hence is $n^{O(1)}$.  Each layer and final
accumulator combines only polynomially many such terms; multiplication by the
fixed-dimensional operator coefficients therefore preserves a polynomial
bound.  Finally $F(n)\leqslant\binom{n^2}{4}=O(n^8)$.  Thus every intermediate
sum and product is $n^{O(1)}$ and has $O(\log n)$ bits.
\end{proof}
\bibliography{references}

\end{document}